\documentclass[11pt]{article}

\newif\ifdraft
\draftfalse

\usepackage{rpmacros}
\usepackage[dvipsnames,svgnames,x11names,hyperref]{xcolor}
\usepackage{stmaryrd}
\usepackage{amsthm}
\usepackage{xfrac}
\usepackage{mathpazo}
\usepackage{gitinfo}
\usepackage{enumitem}
\usepackage{fullpage}
\usepackage{thmtools}
\usepackage{nicefrac}

\usepackage[ruled,vlined,linesnumbered]{algorithm2e}
\DontPrintSemicolon

\usepackage[T1]{fontenc} 
\usepackage[colorlinks]{hyperref}
\hypersetup{
  linkcolor=[rgb]{0,0,0.4},
  citecolor=[rgb]{0, 0.4, 0},
  urlcolor=[rgb]{0.6, 0, 0}
}
\usepackage{lipsum}
\usepackage{calrsfs,bbm}
\usepackage{multirow}
\usepackage{setspace}

\usepackage{microtype}

\usepackage[linesnumbered,vlined,ruled]{algorithm2e}
\usepackage{mathtools}

\usepackage{amsthm}
\usepackage{thmtools,thm-restate}
\usepackage[nameinlink,capitalise]{cleveref}

\numberwithin{equation}{section}
\declaretheoremstyle[bodyfont=\it,qed=\qedsymbol]{noproofstyle}

\declaretheorem[numberlike=equation]{observation}

\declaretheorem[name=Observation,numbered=no]{observation*}

\declaretheorem[numberlike=equation]{fact}

\declaretheorem[numberlike=equation]{theorem}

\declaretheorem[name=Theorem,numbered=no]{theorem*}

\declaretheorem[numberlike=equation]{lemma}
\declaretheorem[name=Lemma,numbered=no]{lemma*}

\declaretheorem[numberlike=equation]{corollary}
\declaretheorem[name=Corollary,numbered=no]{corollary*}

\declaretheorem[numberlike=equation]{proposition}
\declaretheorem[name=Proposition,numbered=no]{proposition*}

\declaretheorem[numberlike=equation]{claim}
\declaretheorem[name=Claim,numbered=no]{claim*}

\declaretheorem[name=Conjecture,numbered=no]{conjecture*}

\declaretheorem[name=Question,numbered=no]{question*}

\declaretheoremstyle[bodyfont=\it,qed=$\lozenge$]{defstyle} 

\declaretheorem[numberlike=equation,style=defstyle]{definition}
\declaretheorem[unnumbered,name=Definition,style=defstyle]{definition*}

\declaretheorem[unnumbered,name=Example,style=defstyle]{example*}

\declaretheorem[unnumbered,name=Notation=defstyle]{notation*}

\declaretheorem[unnumbered,name=Construction,style=defstyle]{construction*}

\declaretheorem[numberlike=equation,style=defstyle]{remark}
\declaretheorem[unnumbered,name=Remark,style=defstyle]{remark*}

\renewcommand{\phi}{\varphi}
\renewcommand{\epsilon}{\varepsilon}

\newcommand{\SP}{\Sigma\Pi}
\newcommand{\SPsize}[1]{(\SP)^k\operatorname{-size}}

\usepackage{nth}
\usepackage{intcalc}
\usepackage{etoolbox}
\usepackage{xstring}

\usepackage{ifpdf}
\ifpdf
\else
\usepackage[quadpoints=false]{hypdvips}
\fi

\newcommand{\ECCCurl}[2]{http://eccc.hpi-web.de/report/\ifnumcomp{#1}{>}{93}{19}{20}#1/#2/}

\newcommand{\shortECCC}[2]{\texttt{\href{http://eccc.hpi-web.de/report/\ifnumcomp{#1}{>}{93}{19}{20}#1/#2/}{eccc:TR#1-#2}}}

\newcommand{\parseECCC}[1]{
\StrSubstitute{#1}{TR}{}[\tmpstring]%
\IfSubStr{\tmpstring}{/}{ 
\StrBefore{\tmpstring}{/}[\ecccyear]%
\StrBehind{\tmpstring}{/}[\ecccreport]%
}{
\StrBefore{\tmpstring}{-}[\ecccyear]%
\StrBehind{\tmpstring}{-}[\ecccreport]%
}%
\shortECCC{\ecccyear}{\ecccreport}}

\newcommand*\samethanks[1][\value{footnote}]{\footnotemark[#1]}

\newcounter{todo}

\ifdraft
\newcommand{\RPnote}[1]{\refstepcounter{todo}\textcolor{WildStrawberry}{\guillemotleft RP: #1 \guillemotright\addcontentsline{tod}{subsection}{[RP]~#1}}}

\newcommand{\gitinfonotecolour}{Gray}
\newcommand{\easteregg}{}
\else
\newcommand{\RPnote}[1]{}
\newcommand{\gitinfonotecolour}{white}
\newcommand{\easteregg}{%
  {\hypersetup{urlcolor=white}%
   \href{https://www.tcs.tifr.res.in/~ramprasad/second-rate-minds/}{Decoding decoding Reed-Solomon codes}}}
\fi

\newcommand{\ignore}[1]{}
\newcommand{\gitinfonote}{git info:~\gitAbbrevHash\;,\;(\gitAuthorIsoDate)\; \;\gitVtag}

\newcommand{\RS}{\operatorname{RS}}
\newcommand{\agr}{\operatorname{agr}}

\allowdisplaybreaks
\title{Reed-Solomon Codes at Capacity: \\
Algorithmic List-Decoding and Proximity Gaps\\[1 ex]
} 
\author{
  {Prahladh Harsha\thanks{Tata Institute of Fundamental Research, Mumbai, India. \href{mailto:prahladh@tifr.res.in}{\{prahladh,} \href{mailto:mrinal@tifr.res.in}{mrinal,} \href{mailto:ramprasad@tifr.res.in}{ramprasad\}@tifr.res.in}. Research supported by the Department of Atomic Energy, Government of India, under project number RTI400112, and in part by research grants from Google and Premji Invest.}}
  \and
  {Mrinal Kumar\samethanks}
  \and
  {Ramprasad Saptharishi\samethanks}
}
\date{}

\begin{document}

\maketitle

\begin{abstract}

     Understanding the limits of list-decodability of Reed-Solomon codes has been one of the most important open problems in algebraic coding theory. Recently, Brakensiek, Chen, Putterman, Zhang, and Zheng, in a remarkable breakthrough, showed that Reed-Solomon (RS) codes over fields of large characteristic are algorithmically list-decodable all the way up to capacity. Building on this result, Jeronimo subsequently extended these techniques to solve the proximity-gaps question for RS codes.

     In this article, we give a unified and transparent exposition of these results. 
\end{abstract}

\section{Introduction}
Reed-Solomon codes~\cite{ReedS1960} are one of the most well-studied families of error-correcting codes. Given a finite field $\F_q$ of size $q=p^r$ and characteristic $p$ and a set $E = \set{\alpha_1,\ldots, \alpha_n}\subseteq \F_q$ of $n$ distinct evaluation points, the code $\RS_q[E,k]$ is the map $\F_q^k \to \F_q^n$ which maps a message $f \in \F_q^k$, viewed as a polynomial $f(X) \in \F_q[X]$ of degree less than $k$, to the evaluation vector $(f(\alpha_1), \ldots, f(\alpha_n)) \in \F_q^n$. This is a linear code with a message space of dimension $k$ and block length $n$, thereby having rate $\rho := \nicefrac{k}{n}$. The code is sometimes identified with the linear space comprising the set of all evaluation vectors. Sometimes we drop the set $E$ and instead write $\RS_q[n,k]$. 

The decoding problem associated with this code is as follows: given a received word $\vecy=(y_1,\ldots,y_n)\in \F_q^n$ and an agreement parameter $A$, find \emph{all} polynomials $f \in \F_q[X]$ of degree less than $k$ whose evaluation vectors agree with the received word $\vecy$ on at least $A$ locations. Determining the agreement threshold at which the list size is bounded, and designing an efficient algorithm to find the resulting list, have been long-standing problems in algebraic coding theory. A related open problem is the proximity-gaps question of RS codes, first studied by Ben-Sasson, Carmon, Ishai, Kopparty and Saraf~\cite{BenSassonCIKS2023}.  Like several other open problems in mathematics and theoretical computer science that have recently met their fate, thanks to some heavy lifting by AI tools, both of these problems were resolved in the last few weeks with AI assistance.

The Proximity Prize (\href{https://proximityprize.org/}{\texttt{https://proximityprize.org/}}) initiated by the Ethereum foundation led to a large crowd-sourced effort to resolve these conjectures. A culmination of these efforts led to the following breakthrough results by Brakensiek, Chen, Putterman, Zhang and Zheng~\cite{BrakensiekCPZZ2026} and Jeronimo~\cite{Jeronimo2026}. First, Brakensiek, Chen, Putterman, Zhang and Zheng established that for all fields of large characteristic, we can in fact decode all the way "up to capacity", i.e., for $A$ as small as $k(1+\epsilon)$ for any $\epsilon >0$. More precisely: 

\begin{theorem}[\cite{BrakensiekCPZZ2026}]
  \label{thm:BCPZZ}
  Let $\rho, \epsilon > 0$, and let $k, n$ be large enough positive integers with $\nicefrac{k}{n} \geq \rho$.
  Let $\F_q$ be any finite field of size $q=p^r$ for prime $p$ satisfying $q \geq n$ and $p > k$. Then, there is an algorithm $\mathcal{A}(\epsilon)$ that on input a list $(\alpha_1, \alpha_2, \dots, \alpha_n)$ of $n$ distinct elements of $\F_q$ and a received word $\vecy = (y_1,\ldots, y_n) \in \F_q^n$, outputs \emph{all} polynomials $f \in \F_q[X]$ of degree less than $k$ satisfying $\abs{\agr(f,\vecy)} \geq k(1+\epsilon)$ where 
  \[
  \agr(f, \vecy) := \setdef{i\in [n]}{f(\alpha_i) = y_i}
  \]
  in time $q^{C(\epsilon, \rho)}$ where $C$ is some function of $\epsilon$ and the rate $\nicefrac{k}{n}$\footnote{The original writeup of \cite{BrakensiekCPZZ2026} proved this result only for low rates and this was then extended to all rates in a subsequent revision using a reduction due to Alrabiah, Goyal and Guruswami. The proof presented in this exposition directly works for all rates and all sets of evaluation points and  avoids the use of the Alrabiah-Goyal-Guruswami reduction (we note the revised version of \cite{BrakensiekCPZZ2026} already mentions, without proof, that such a simplification is feasible).}. 
\end{theorem}

Almost immediately after this result was posted, the techniques developed in this paper were extended by Jeronimo~\cite{Jeronimo2026} to solve the proximity-gaps problem for RS codes. We state the result for lines for simplicity (although it readily extends to the more general statement). 

\begin{theorem}[\cite{Jeronimo2026}]
  \label{thm:Jeronimo}
  Let $\rho, \epsilon > 0$, and let $k, n$ be large enough positive integers with $\nicefrac{k}{n} \geq \rho$.
  Let $\RS_q[n,k]$ be a Reed-Solomon code over a finite field $\F_q$ of characteristic $p$ and rate at least $\rho$ for any set $E = \{\alpha_1, \dots, \alpha_n\} \subseteq \F_q$ of evaluation points satisfying $q \geq n$ and $p > k$. There is a bound $L = n^{C(\epsilon, \rho)}$, where $C$ is some function of $\epsilon$ and $\rho$, such that the following holds for every pair of functions  $u, v \colon E \to \F_q$:
  \begin{quote}
  Suppose for at least $L$ values of $\zeta \in \F_q$ there exists some polynomial $f_\zeta$ of degree less than $k$ such that $\abs{\agr(u + \zeta v, f_{\zeta})} \geq k(1 + \epsilon)$. Then, there are polynomials $g,h$ of degree less than $k$ such that 
  \[
  \abs{\setdef{i \in [n]}{u(\alpha_i) = g(\alpha_i) \text{ and }v(\alpha_i) = h(\alpha_i)}} \geq k(1 + \epsilon)\, . 
  \]
  I.e., there is a set of size at least $k(1+\epsilon)$ of \emph{correlated agreements} of $u$ and $v$ with $\RS_q[n,k]$, thereby showing that $u + \zeta v$ has $k(1+\epsilon)$ agreements with $\RS_q[n,k]$ \emph{for every} $\zeta \in \F_q$. 
  \end{quote}
\end{theorem}

Thus, if the field size $q$ is a large enough polynomial in $n$ (more precisely, $q \gg n^{C(\epsilon, \rho)}$), a very small fraction of $\zeta$'s that witness closeness of $u + \zeta v$ with $\RS_q[n,k]$ is sufficient to conclude that $u+ \zeta v$ is close to $\RS_q[n,k]$ for \emph{every} choice of $\zeta$.

In this note, we give an exposition of the proofs of both these results. On the way, we highlight the key ideas and make slight improvements in some of the parameters involved. The exposition is largely self-contained, barring a couple of standard facts from algebraic geometry and coding theory literature, which we state explicitly. 

\subsection{A birds-eye view of the proofs}
For the benefit of readers familiar with the existing algebraic coding theory literature, we summarize below the high-level ideas that go into the proof of these statements. Others may prefer to skip this overview and proceed directly to the subsequent sections, which explain the proofs from first principles.

\subsection*{Algorithmic list-decoding (\cite{BrakensiekCPZZ2026})}

The \cite{BrakensiekCPZZ2026} algorithm follows the by now standard template of Sudan and Guruswami-Sudan~\cite{Sudan1997,GuruswamiS1999} for list-decoding: explain the received using a non-zero polynomial $Q$ and extract all close-enough  codewords from this polynomial. The existence of such a non-zero polynomial $Q$ is shown by setting a linear system and showing that the number of variables in this system exceeds the number of linear constraints. 

The \cite{BrakensiekCPZZ2026} algorithm deviates from the Guruswami-Sudan algorithm in multiple important aspects. Some of these new technical ingredients in the algorithm in \cite{BrakensiekCPZZ2026} can be traced back to the algorithm of Kopparty \cite{Kopparty2015} for list-decoding univariate multiplicity codes, in both implicit and explicit ways. We start by recalling the outline of Kopparty's algorithm and then discuss how some of these ideas are adapted to the setting of Reed-Solomon codes in \cite{BrakensiekCPZZ2026}. The technical insights in \cite{BrakensiekCPZZ2026} that make this adaptation feasible are perhaps of independent interest in the setting of the polynomial method, beyond this specific proof. 

\paragraph{Kopparty's univariate multiplicity code decoder: } For list-decoding univariate multiplicity codes (over fields of sufficiently large characteristic), Kopparty suggested obtaining a non-zero multivariate polynomial $Q(X,Y_0, \dots, Y_d)$ of low enough degree such that for every message polynomial $f$, if $f$ agrees with the received word at evaluation point $\alpha_i$, then 
\begin{equation*}
Q(\alpha_i + T, f(\alpha_i + T), f^{(1)}(\alpha_i + T), \ldots, f^{(d)}(\alpha_i + T)) = 0 \bmod T^m\,.
\end{equation*}
This ensures that if a low degree $f$ has sufficiently large agreement with the received word, then $f$ satisfies the differential equation 
\[
Q(X,f(X), f^{(1)}(X), \dots, f^{(d)}(X)) \equiv 0.
\] 
Thus, the task of list-decoding has reduced to solving the above differential equation. The final ingredient of Kopparty's decoder is an efficient algorithm that solves the above differential equation to obtain a list of degree $k$ univariates of size at most $q^{O(d)}$ containing all close enough message polynomials, where $q$ is the size of the underlying finite field. For a constant $d$, this gives a polynomial size list for decoding univariate multiplicity codes up to capacity over finite fields of sufficiently large characteristic.

We now discuss the main steps in \cite{BrakensiekCPZZ2026}. 

\paragraph{Step 1 (multivariate interpolation):}
The first idea in \cite{BrakensiekCPZZ2026} is to construct a low degree \emph{multivariate} polynomial $Q$ in the interpolation step. The goal again is to obtain a non-zero $Q(X,Y_0, \dots, Y_d)$ such that if $f(X)$ and the received word agree at coordinate $\alpha_i$, then 
\[
Q(\alpha_i + T, f(\alpha_i + T), f^{(1)}(\alpha_i + T), \ldots, f^{(d)}(\alpha_i + T)) = 0 \bmod T^m\,.
\]
Note that this is equivalent to the polynomial $Q(X,f(X),\dots,f^{(d)}(X))$ having  a zero of multiplicity at least $m$ at $\alpha_i$. While similar to the interpolation step in \cite{Kopparty2015}, the absence of high multiplicity evaluations of the message polynomial in the setting of Reed-Solomon codes leads to obvious technical issues. In the absence of access to high multiplicity evaluations,  \cite{BrakensiekCPZZ2026} use formal indeterminates $Y_1,Y_2,\ldots, Y_d$ for these high-order derivatives. However, it can be checked (with some effort) that treating the high multiplicity information as independent formal variables, with no further apparent algebraic relationship between them and the received word, is in itself insufficient to list decode beyond the Johnson radius. 

\paragraph{Step 2 (an algebraic relation between derivatives): }
The second ingredient of \cite{BrakensiekCPZZ2026} that gives a non-trivial way to use the first idea discussed above, is that the high-order derivatives and the evaluation indeed satisfy a useful algebraic relationship. This is given by the "reverse" Taylor expansion:
\begin{align*}f(X+T) & = f(X) + T f^{(1)}(X+T) + \cdots + (-1)^{d+1} T^d f^{(d)}(X+T) + E
\end{align*}
where the error term $E = E(X,T)$ is divisible by $T^{(d+1)}$. 
This relationship is then incorporated  to bound the number of variables and constraints as follows:
\begin{align*}
\text{\# variables} \geq c_{\epsilon} \cdot \frac{m^2A^2}{k} \cdot \abs{\calB} \,,
&& \text{\# constraints} \leq n\cdot \frac{m^2}{d} \cdot \abs{\calB + \Lambda(m)}
\end{align*}
Here, $n$ is the blocklength, $\epsilon$ measures how close we are to capacity (i.e., we are decoding from agreements $A \geq k(1+\epsilon)$), $m$ is the multiplicity parameter (as in the Guruswami-Sudan algorithm), $d$ the number of high-order derivatives used and $c_{\epsilon}$ is a constant only dependent on $\epsilon$. The only new terms here are the sets $\calB$ and $\Lambda(m)$. The set $\calB$ refers to the set of exponent vectors $(b_1,b_2,\ldots,b_d)$ allowed for monomials of the form $Y_1^{b_1}Y_2^{b_2}\cdots Y_d^{b_d}$ in the polynomial $Q$ while $\Lambda(m)$ is the set of non-negative integer vectors $(b'_1,b'_2,\ldots,b'_d)$ satisfying $\sum i b'_i \leq m$. 

The key point to note is the following: for the number of variables to be larger than the number of constraints, it suffices to have 
\[
A > \left(\frac{1}{c_{\epsilon}}\cdot \frac{nk}{d} \cdot \frac{\abs{\calB+\Lambda(m)}}{\abs{\calB}}\right)^{1/2} = k\left(\frac{1}{c_{\epsilon}}\cdot \frac{1}{\rho d} \cdot \frac{\abs{\calB+\Lambda(m)}}{\abs{\calB}}\right)^{1/2}, 
\]
where the second expression is obtained by setting $n = k/\rho$, where $\rho$ is the rate of the code. 

Thus, to minimize $A$, ideally to $A = (1+\epsilon)k$, parameters need to be set so that the expression $\left(\frac{1}{c_{\epsilon}}\cdot \frac{1}{\rho d} \cdot \frac{\abs{\calB+\Lambda(m)}}{\abs{\calB}}\right)^{1/2}$ is less than $(1+\epsilon)$. It turns out that this requires a careful choice of $\calB$ as we discuss below. 

\paragraph{Step 3 (careful choice of a monomial basis for $Q$): }
A natural choice for $\calB$ is the set $(b_1,b_2,\ldots,b_d)$ in $ \Z_{\geq 0}^d$ of vectors satisfying $\sum b_i \leq D$. This corresponds to interpolating a $Q$ whose total degree in the variables $Y_1, \dots, Y_d$ is at most $D$ and is the typical space of polynomials that one works with in many applications of the polynomial method, including those in algebraic coding theory. Unfortunately this does not appear to work here,  as for this choice of $\calB$, the corresponding ratio $\frac{\abs{\calB+\Lambda(m)}}{\abs{\calB}}$ is $\Omega(d)$, and the lower bound on $A$ obtained is much larger than the desired bound of $(1+\epsilon)k$. 

The careful choice of $\calB$ is the third main ingredient of the proof in \cite{BrakensiekCPZZ2026}. This choice ensures that the ratio $\frac{\abs{\calB+\Lambda(m)}}{\abs{\calB}}$ is less than $d^{1-\eta}$ for some constant $\eta \in (0,1)$. By a careful setting $d$, depending on the rate $\rho$ and $\epsilon$, the bound on $A$ obtained at the end \footnote{We also remark that the bound obtained on $A$ can never be less than $k$, as one would expect. This is evident from the details in the proof. } indeed looks like $(1+\epsilon)k$.   Choosing $\calB = \Lambda(W)$ for an appropriate integer $W$ works (the choice $W \approx md(1+\epsilon/2)/\ln d$ suffices) since in this case $\calB+\Lambda(m)= \Lambda(W+m)$ and we can exactly estimate the ratio $\frac{\abs{\Lambda(W+m)}}{\abs{\Lambda(W)}}$ (for $m$ sufficiently large compared to $d$).

\paragraph{Step 4 (differential equation solver): }
The non-zero polynomial $Q$ obtained at the end of the interpolation step is then shown to contain all close enough message polynomials $f$ in its belly in the same sense as that in \cite{Kopparty2015}, and Kopparty's differential equation solver is then used off-the-shelf to recover these messages. The parameter $d$ has to be set to a constant to that the list size is polynomially bounded and the differential equation solver runs in polynomial time. 

\subsection*{Proximity gaps for Reed-Solomon codes  \cite{Jeronimo2026}}

The proximity-gaps problem was  first studied by Ben-Sasson, Carmon, Ishai, Kopparty, and Saraf~\cite{BenSassonCIKS2023}, who proved results when the number of agreements $A$ is at least the unique-decoding radius $\nicefrac{(n+k)}{2}$ or the Johnson radius $\sqrt{nk}$. In both these settings, \cite{BenSassonCIKS2023} used the corresponding decoding algorithms, in other words the Welch-Berlekamp \cite{WelchB1986}-algorithm (as abstracted by Gemmell and Sudan \cite{GemmellS1992}) in the unique-decoding setting and Guruswami-Sudan \cite{GuruswamiS1999}-algorithm in the list-decoding up to Johnson radius setting. Given that we now have a list-decoding algorithm all the way up to capacity, \`a la \cite{BrakensiekCPZZ2026}, we can naturally ask if we can further reduce the number of agreements in the proximity gap question to $k(1+\epsilon)$. And this is precisely what \cite{Jeronimo2026}'s result does. However, \cite{Jeronimo2026}'s result does not follow the approach of \cite{BenSassonCIKS2023}. It instead proceeds as follows. It observes that a pruning algorithm similar to that of Kopparty, Ron Zewi, Saraf and Wootters \cite{KoppartyRSW2023} and Tamo \cite{Tamo2024} can be used to bound the number of polynomials in the list by a polynomial in $n$, whose exponent depends only on the rate $\rho$ and $\epsilon$ and is independent of the underlying field size. The pruning here is more involved compared to \cite{KoppartyRSW2023,Tamo2024} as the objects involved are varieties instead of affine subspaces and this is where we need some basic algebraic geometry facts. Given such a pruning method, one can bound the number of values of $\zeta$ for which there is no coordinated agreement. 

\paragraph{AI disclosure:} The technical details in both the papers \cite{BrakensiekCPZZ2026,Jeronimo2026} (especially the latter) are very dense and understanding these papers (for us) required several extended conversations with Claude (Opus and Sonnet). The proof here follows the same strategy as the original papers, but the details are (in our opinion) much easier to follow. This exposition is entirely written by the human authors. 

\paragraph{Human disclosure: } Many thanks to Irit Dinur, Swastik Kopparty and Madhu Sudan for substantial and extremely helpful comments on various parts of the paper. These comments helped us improve the presentation significantly. 

A special thanks to Swastik for a super careful reading of the technical sections and pointing us to some references that we had missed. 

\section{Notation and preliminaries}

We will use boldface letters such as $\vecy$ to denote a tuple $(y_1,\ldots, y_k)$;  the number of coordinates would usually be clear from context. 

\paragraph{Hasse Derivatives:} For a polynomial $f(X)$ of degree at most $d$, we define the \emph{Hasse derivative} of $f$, denoted by $f^{(0)} = f$, $f^{(1)}$, etc. as the terms of the Taylor expansion of $f(X + T)$ for a formal variable $T$:
\[
f(X + T) = f^{(0)}(X) + f^{(1)}(X) \cdot T + f^{(2)}(X) \cdot T^2 + \cdots + f^{(d)}(X) \cdot T^d. 
\]
 While working with finite fields, we follow the customary procedure of using Hasse derivatives instead of standard derivatives. However, since the results presented in the article apply only to fields of large characteristic, we might have as well used standard derivatives. In this setting, the Hasse derivatives $f^{(i)}$ are related to the standard derivatives $\frac{d^i f}{dX^i}$ via the formula $f^{(i)}(X) = \frac{1}{i!} \cdot \frac{d^i f}{dX^i}$.

A polynomial $f(X)$ is said to have a root $\alpha$ of multiplicity $m$ if $(X - \alpha)^m$ divides $f$. Equivalently, $f^{(i)}(\alpha) = 0$ for all $0 \leq i < m$. 

Here is a standard combinatorial fact that will be used in the algorithm. 

\begin{fact}[number of bivariate monomials with weighted-degree bound]
  \label{fact:weighted-degree-mons}
  For positive integers $k,D$, if $N(k,D)$ refers to the number of monomials of the form $X^a Y^b$ with $a + kb < D$, then 
  \[
  \frac{D(D+k)}{2k} \leq N(k,D) \leq \frac{D(D+k)}{2k} + \frac{k}{8}
  \]
\end{fact}

For a received word $\vecy \in \F_q^n$ and polynomial $f(x)$, define $\agr(f, \vecy)$ to be 
\[
\agr(f, \vecy) := \setdef{i \in [n]}{f(\alpha_i) = y_i}. 
\]

\section{Sudan's template for algebraic decoding algorithms}

The history of decoding algorithms for Reed-Solomon codes is quite rich with algorithms known since as far back as 1960s. We only mention a subset of these results, with the focus being on results that are most relevant to the recent progress on this problem. 

Berlekamp-Welch~\cite{WelchB1986} (as abstracted by Gemmell and Sudan~\cite{GemmellS1992}) gave an efficient \emph{unique-decoding} algorithm where the number of agrements $A \geq \tfrac{n+k}{2}$. Sudan~\cite{Sudan1997} gave an efficient \emph{list-decoding} algorithm with $A \gtrsim \sqrt{2nk}$, and Guruswami and Sudan~\cite{GuruswamiS1999} further improved this to $A \gtrsim \sqrt{nk}$ by adding the \emph{method of multiplicities}. This bound of $\sqrt{nk}$ agreements is also known as the Johnson radius, and \emph{any} code has small list size when dealing with $\gtrsim \sqrt{nk}$ agreements (note this is just a combinatorial guarantee about \emph{any} code while the Guruswami-Sudan algorithm yields \emph{efficient} list-decoding algorithms up to this radius for RS codes). 

To go beyond the Johnson radius, variants of Reed-Solomon codes were devised (such as multiplicity codes, or Folded Reed-Solomon codes \cite{ParvareshV2005,GuruswamiR2008,GuruswamiW2013,Kopparty2015}) that are efficiently list-decodable up to capacity (i.e., $A \geq k(1+\epsilon)$ for any $\epsilon > 0$). \\

All the above algorithms proceed via a very similar template, first formalized for unique-decoding by Gemmell and Sudan~\cite{GemmellS1992} and for list-decoding by Sudan~\cite{Sudan1997}: Given a received word $\vecy = (y_1,\ldots, y_n)$, 
\begin{itemize}\itemsep 0pt
  \item Build an ``explainer polynomial'' $Q$ based on the data $(y_1,\ldots, y_n)$. 
  \item Extract from $Q$ all $f$ such that $\agr(f, \vecy) \geq A$. 
\end{itemize}
The algorithm of \cite{BrakensiekCPZZ2026} proceeds along the same template, but with an additional ingredient of \emph{hidden derivatives}. In its final form, this idea lets them adapt the techniques of Kopparty's univariate multiplicity code decoder in \cite{Kopparty2015} to the setting of Reed-Solomon codes.  To illustrate the techniques in \cite{BrakensiekCPZZ2026} better, it would be useful to first revisit the method of multiplicities used in the Guruswami-Sudan~\cite{GuruswamiS1999} algorithm. 

\subsection{Guruswami-Sudan's list-decoding algorithm}

Let us fix a parameter $m$ (for the multiplicity). 

Given a received word $\vecy = (y_1,\ldots, y_n)$, we try to find a nonzero bivariate polynomial $Q(X,Y)$ with the following properties:
\begin{description}
\item{Monomial space:} Every monomial $X^a Y^r$ in $Q$ satisfies $a + kr < mA$. In other words, the $(1,k)$-degree of $Q$ is strictly smaller than $mA$. 
\item{Constraints:} For every $i \in [n]$, consider the polynomial 
\[
Q(\alpha_i + T, y_i + U) = \sum_{a, r} q_{i,a,r} T^a U^r
\]
and add the constraint that $q_{i,a,r} = 0$ for all $a,r$ satisfying $a + r < m$. 
\end{description}

\begin{claim}
Suppose such a nonzero $Q(X,Y)$ exists. If $f$ is a polynomial of degree less than $k$ such that $\abs{\agr(f, \vecy)} \geq A$, then $R_f(X):=Q(X, f(X)) \equiv 0$ or equivalently $Y-f(X)$ is a factor of $Q$.
\end{claim}
\begin{proof}
Since $Q$ has $(1,k)$-degree less than $mA$, the polynomial $R_f(X)=Q(X,f(X))$ has degree strictly smaller than $mA$. 

Suppose $f(\alpha_i) = y_i$ for some $i$, then $f(\alpha_i + T) = y_i + U_{f,i}(T)$ for some polynomial $U_{f,i}$ that is divisible by $T$. Therefore, 
\begin{align*}
R_f(\alpha_i + T) &= Q(\alpha_i + T, f(\alpha_i + T))\\
& = Q(\alpha_i + T, y_i + U_{f,i}(T))\\
& = 0 \bmod T^m \quad (\text{from our constraints})
\end{align*}
Thus, $R_f$ has $\alpha_i$ as a root of multiplicity at least $m$. But if $\agr(f, \vecy) \geq A$, then get that $R_f$ has $A$ roots of multiplicity at least $m$ each but $\deg R_f < mA$. Hence, $R_f(X) \equiv 0$. 
\end{proof}

Once we find such a nonzero $Q$, the list-decoding algorithm is now immediate; just factor $Q(X,Y)$ and any $f$ in the list would correspond to a factor of the form $Y - f(X)$ of $Q(X,Y)$. 

Thus, the parameter choices are basically to ensure that such a nonzero polynomial exists. To this end, it suffices to show that the number of constraints is less than the number of indeterminates in $Q$. 

\paragraph{Counting indeterminates in $Q$:} We need to count the number of monomials $X^a Y^r$ such that $a + kr < mA$. When $m$ is large enough, it can be shown (\cref{fact:weighted-degree-mons}) that this is roughly
\[
\frac{m^2 A^2}{2k}.
\]

\paragraph{Counting constraints:} The number of constraints per $i \in [n]$ is the number of monomials of the form $T^a U^r$ with $a + r < m$, and this is about $\binom{m}{2}$. 

\noindent Hence, to ensure that there is a nonzero $Q$, it suffices to ensure that 
\begin{align*}
\frac{m^2A^2}{2k} & > n \cdot \frac{m^2}{2}
\end{align*}
and this is satisfied if  $A \gtrsim \sqrt{nk}$. This completes the description of the Guruswami-Sudan algorithm.

\subsection{Using additional derivatives}

A very natural idea is to use some of the techniques that were useful in the setting of multiplicity codes or Folded Reed-Solomon codes, namely to use higher derivative information. Suppose we could, as first suggested by Kopparty~\cite{Kopparty2015} while decoding univariate multiplicity codes,  find a nonzero polynomial $Q(X, Y_0, \ldots, Y_d)$ with $(1, k, k, \ldots, k)$-degree\footnote{Technically, one should work with $(1,k,k-1,\ldots,k-d)$-degree, but the simpler $(1,k,k,\ldots,k)$-degree bound suffices.} less than $mA$ with the property that whenever $f(\alpha_i) = y_i$ we have 
\begin{equation}\label{eqn:hidden-derivative}
Q(\alpha_i + T, f(\alpha_i + T), f^{(1)}(\alpha_i + T), \ldots, f^{(d)}(\alpha_i + T)) = 0 \bmod T^m
\end{equation}
where $f^{(i)}(X)$ is the $d$-th Hasse derivative of $f$. 

Then, if $f$ has at least $A$ agreements, then we would again have that 
\[
R_f(X) := Q(X, f(X), f^{(1)}(X), \ldots, f^{(d)}(X)) = 0.
\]
That is, $f$ satisfies the ``differential equation'' given by $Q$. The intuition behind working with a $Q$ that is multivariate (and not just bivariate) is that the dimension of the space of such polynomials at any fixed weighted degree grows as the number of variables increases, and hence there is a chance that there would be a non-zero $Q$ of lower overall degree that satisfies the interpolation constraints.  

Fortunately, once we have such a $Q$, we also know how to find all such $f$'s that satisfy $R_f(X) \equiv 0$ via a differential equation solver in Kopparty's~\cite{Kopparty2015} decoding algorithm for univariate multiplicity codes. 

\begin{theorem}[\cite{Kopparty2015}]
  \label{thm:kopparty-algorithm}
  Suppose $Q(X, Y_0, \cdots, Y_d) \in \F_q[X, Y_0,\ldots, Y_d]$ is a nonzero polynomial with $\operatorname{char}(\F_q) = p > k, \deg_{Y_i}(Q)$. Then the set of all polynomials $f(X)$ of degree less than $k$ that satisfy
  \[
  Q(X, f(X), f^{(1)}(X), \cdots, f^{(d)}(X)) = 0
  \]
  can be enumerated in time $q^{O(d)} \cdot\poly(k, d_Q)$, where $d_Q$ is the $(1,k,k,\cdots, k)$-degree of $Q$. 
\end{theorem}

\begin{remark*}
The conditions above are slightly different from \cite{Kopparty2015} but follow from the proof. 
\end{remark*}

\noindent The main question is how we should be setting up our constraints to ensure \eqref{eqn:hidden-derivative}. 

\subsubsection*{Na\"ive attempt}

Let us try to follow the Guruswami-Sudan outline. We know that if $f(\alpha_i) = y_i$, then $f(\alpha_i + T) = y_i + U$ for some polynomial $U$ divisible by $T$. However, we have no information about the higher order derivatives. Thus, we may attempt to leave them as indeterminates. 

\paragraph{Constraints:}
For each $i \in [n]$, consider the polynomial 
\[
Q(\alpha_i + T, y_i + U, Y_1, \ldots, Y_d) = \sum q_{i,a,r,(b_1,\ldots,b_d)} T^a U^r Y_1^{b_1} \cdots Y_d^{b_d} = \sum q_{i,a,r,\vecb} T^a U^r \vecY^{\vecb}
\]
and insist that $q_{i,a,r, \vecb} = 0$ whenever $a + r < m$. 

Thus, if $\calB$ refers to the space of ``valid'' $\vecb=(b_1,\ldots, b_d)$, then this amounts to roughly $n \cdot \binom{m+1}{2} \cdot \abs{\calB}$ constraints overall. 

\paragraph{Number of indeterminates:}
We wish to get an estimate of the dimension of the interpolation space, namely the number of monomials we could have in $Q$. Thus, we want to count $X^a Y_0^{b_0} Y_1^{b_1} \cdots Y_d^{b_d}$. To make it easier to relate this with $\calB$, let us assume that we $\abs{\vecb}=b_1 + \cdots + b_d < 0.99 \cdot \frac{mA}{k}$, so that $a + k b_0$ could be as large as $0.01 \cdot mA$ no matter what $\vecb=(b_1,\ldots, b_d)$ is. Combining with \cref{fact:weighted-degree-mons}, we get that the number of indeterminates in $Q$ is at least 
\[
\frac{(0.01\cdot  mA)^2}{2k} \cdot \abs{\calB}\ . 
\]

Observe that we actually made things worse by looking at higher derivatives: our interpolation space seemed to have shrunk (superficially at least from the above bound) and we didn't get any useful gain from the constraints. They both got scaled by a factor of $\abs{\calB}$ and feels like we might as well have forced $\vecb= (b_1,\ldots, b_d) = \mathbf{0}$ (thus reducing to Guruswami-Sudan). 

\subsection{The key insight: hidden derivatives}

The key insight of \cite{BrakensiekCPZZ2026} is that there is a relation between the evaluation $y_i$ and the higher derivatives given by the following "reverse" Taylor expansion: 
\begin{align*}
f(X) & = f((X + T) - T) \\
     & = f(X + T) + (-T) \cdot f^{(1)}(X + T) + \cdots + (-T)^d f^{(d)}(X + T) \bmod T^{d+1}\\
\therefore f(X+T) & = f(X) + T f^{(1)}(X+T) + \cdots + (-1)^{d+1} T^d f^{(d)}(X+T) + E
\end{align*}
where the error term $E=E(X,T) = 0 \bmod T^{d+1}$. 

To illustrate the usefulness of this, let us focus on the case when $d = 1$, namely we introduce one hidden derivative. In this scenario, 
\[
f(\alpha_i + T) = y_i + T f^{(1)}(\alpha_i + T) + E
\]
where $E = 0 \bmod T^2$. We will see that this change is already enough to go past the Johnson radius. 

\subsubsection{Breaching the Johnson radius}

Let $A$ be the number of agreements we are targeting (and we just want to go below $\sqrt{nk}$ for now). Gabizon~\cite{Gabizon2026} studied the $d=1$ setting of the \cite{BrakensiekCPZZ2026} where one can go past the Johnson radius with a single hidden derivative, and also provided a linear list size in this setting (using ideas from \cite{Jeronimo2026}). 

\paragraph{Interpolation space:} Fix a large constant $L$. Let us find a polynomial $Q(X, Y_0, Y_1) = \sum q_{a,r,b} X^a Y_0^r Y_1^b$ such that
\begin{itemize}\itemsep 0pt
\item $a + kr + kb < mA$,
\item $b < Lm$. 
\end{itemize}

The dimension of this space is at least 
\begin{align*}
\sum_{b = 0}^{Lm-1} \sum_{\substack{a,r\\a + kr < ma - kb}} 1 & \geq Lm \cdot \frac{(mA - k m L)^2}{2k} = \frac{L m^3}{2k} \cdot (A - kL)^2
\end{align*}

\paragraph{Constraints: } For each $i \in [n]$, consider the polynomial 
\[
Q_i(T, E, Y_1) = Q(\alpha_i + T, y_i + TY_1 + E, Y_1) = \sum_{a,r,b} q_{i,a,r,b} T^a E^r Y_1^c. 
\]
We will add a constraint forcing $q_{i,a,r,c} = 0$ for every $(a,r,c)$ satisfying $a + 2r < m$. 

Note that for monomials $T^a E^b Y_1^c$ with $a < m$, we must have that $c < Lm + m$ since any contribution of $Y_1^{\ell}$ from $(y_i + TY_1 + E)^r$ comes accompanied by $T^{\ell}$ as well and the $Y_1$ degree of $Q$ itself is at most $Lm$. Thus, the total number of constraints (again using \cref{fact:weighted-degree-mons}) is at most 
\begin{align*}
n \cdot \sum_{c = 0}^{Lm + m - 1} \sum_{a + 2r < m} 1 \leq n \cdot m(L+1) \cdot \frac{m^2}{4}
\end{align*}

Thus, we would have a nonzero $Q$ if the number of constraints is smaller than the number of indeterminates, namely
\begin{align*}
\frac{L m^3}{2k} \cdot (A - kL)^2 & > n \cdot m(L+1) \cdot \frac{m^2}{4}.
\end{align*}
Thus, suffices to have $(A - kL)^2  > \frac{nk}{2} \cdot \inparen{1 + \frac{1}{L}}$. Or, equivalently, \[
A  > kL + \left(\frac{nk}{2} \cdot \inparen{1 + \frac{1}{L}}\right)^{1/2}.
\]

We already see that we are asymtotically approaching $A \gtrsim \sqrt{nk/2}$ by setting $L$ to be a large constant (at least when the rate is small), even with a single hidden derivative. By setting appropriately large number of derivatives, we will be able to further reduce the number of agreements $A = k (1 + \epsilon)$ for an arbitrarily small $\epsilon$.

\subsection{High-level picture for the general case}

Like in the illustrated case above, we begin with $Q(X, Y_0,\ldots, Y_d)$ and add linear constraints to try and ensure that 
\[
Q(\alpha_i + T, y_i + TY_1 - T^2 Y_2 + \cdots (-1)^{d+1} T^d Y_d + E, Y_1, \ldots, Y_d) = 0 \bmod{T^m}
\]
while keeping in mind that we expect $E$ to be $0 \bmod T^{d+1}$. In other words, if we expand the above into monomials of the form $T^a E^r Y_1^{c_1} \cdots Y_d^{c_d}$, we will expect the coefficient of all monomials with $a + (d+1)r < m$ to be zero. Thus, if $f(y_i) = \alpha_i$ and 
\[
R_{f}(\alpha_i + T) = Q(\alpha_i + T, f(\alpha_i + T), \cdots, f^{(d)}(\alpha_i + T)),
\]
we observe that this is divisible by $T^{m}$ by substituting $Y_j = f^{(j)}(\alpha_i + T)$ and 
\[
E = f(\alpha_i + T) - y_i - \sum_{j=1}^{d} (-1)^{d+1} T^j f^{(j)}(\alpha_i + T)
\]

To count the constraints, we need to count the number of monomials of the form $T^a E^r Y_1^{b_1} \cdots Y_d^{b_d}$ with $a + (d+1)r < m$. Focusing on the number of possibilities for $(a,r)$, we get $\frac{m^2}{2(d+1)}$ instead of the na\"ive bound of $\frac{m^2}{2}$. Hence, we have somehow gotten a saving of a factor of $d+1$. 

But we have lost something as well. The set of exponents $(c_1,\ldots, c_d)$ for $Y_1,\ldots, Y_d$ is now no longer just an element of $\calB$ as we have a contribution from the substitution we made. In fact, any $Y_1^{c_1}\cdots Y_d^{c_d}$ that we obtain can be written as $Y_1^{b_1} \cdots Y_d^{b_d} \cdot Y_1^{b_1'} \cdots Y_d^{b_d'}$ where $(b_1,\ldots, b_d) \in \calB$ and $Y_1^{b_1'}\cdots Y_d^{b_d'}$ is the additional degree we obtain from the substitution. However, every monomial of the form $Y_1^{b_1'}\cdots Y_d^{b_d'}$ from $(y_i + TY_1 - T^2 Y_2 + \cdots \pm T^d Y_d + E)^r$ comes accompanied by a $T^{\sum i b_i'}$. Since the constraints we are imposing are only on monomials with $T$-degree less than $m$, we have that $(c_1,\ldots, c_d) \in \calB + \Lambda(m)$ where 
\[
\Lambda(m) := \setdef{(b_1',\ldots, b_d') \in \Z_{\geq 0}^d}{\sum i b_i' \leq m}. 
\]
Thus, the number of constraints is at most $n \cdot \frac{m^2}{2(d+1)} \cdot \abs{\calB+\Lambda(m)}$. 

From the counting done in the naive attempt in the previous section, the number of coefficients in $Q$ on the other hand is at least  $(0.01 mA)^2\abs{\calB}/2k$, which, for $A = k(1+\epsilon)$ is at least $\left(0.00005(1+\epsilon)^2 m^2k\abs{\calB}\right) = c_\epsilon m^2 k \abs{\calB}$ for some constant $c_\epsilon$. 

Hence, if the factor of $(d+1) = \Omega(d)$ can offset $\frac{1}{c_{\epsilon}}\cdot \frac{n\abs{\calB+\Lambda(m)}}{k\abs{\calB}}$, then, we would have fewer constraints than indeterminates and we can find our polynomial $Q$. 

\paragraph{Getting a sense of the sizes:} We are hoping to work with $A = k(1 + \epsilon)$ for some small constant $\epsilon > 0$ and hence $mA/k = m(1 + \epsilon)$. If $\calB$ is chosen to be $\calB_d := \setdef{(b_1,\ldots, b_d)}{\sum b_i \leq D := m(1 + \epsilon)}$, we have that $|\calB_d| \approx \tfrac{D^d}{d!}$ . However, note that $\calB + \Lambda(m)$ certainly contains points of the form $\setdef{(a, b_2,\ldots, b_d)}{0 \leq a \leq m\;,\; b_2 + \cdots + b_d \leq D}$ and this already accounts for about 
\[
m \cdot \frac{D^{d-1}}{(d-1)!} = \Omega(d \cdot |\calB_d|)
\]
since $D \approx m$. Therefore, this completely offsets the gain we obtained earlier. 

Hence, we need to choose a suitable $\calB$ so that $\calB + \Lambda(m)$ is not that much larger than $\calB$ itself, and turns out a convenient choice is $\calB = \Lambda(W)$ for a suitably chosen $W$. Then, $\calB + \Lambda(W) = \Lambda(W) + \Lambda(m) = \Lambda(W + m)$ and
\[\frac{\abs{\calB+\Lambda(m)}}{\abs{\calB}} =
\frac{\abs{\Lambda(W+m)}}{\abs{\Lambda(W)}} \approx \frac{(W+m)^d}{W^d}
\]
If $W = \frac{md (1 + \delta)}{\ln d}$ for some suitable $0 < \delta < \epsilon$, the above ratio becomes $d^{1/(1 + \delta)}$ which can be offset by the factor $d$ we had obtained above. While lower bounding the dimension of the interpolation space, we had used the rough estimate that $\abs{\vecb} = \sum b_i$ is bounded from above by $0.99 \cdot \frac{mA}{k}$. This argument can be made precise if we had a bound on $\E_{\vecb \in \Lambda(W)}[\sum b_i]$ instead.

Thus, much of the care in the proof is in this choice of $W$, understanding the size of $\Lambda(W)$ and $\E_{\vecb \in \Lambda(W)}[\sum b_i]$. We do that in the next section and defer the detailed calculations to the subsequent section. 

\section{The simplex $\Lambda(W)$}

We will mostly be focusing on the discrete set of points $\Lambda(W)$ defined earlier, and its continuous analogue:
\begin{align*}
\Lambda(W) &:= \setdef{(b_1,\ldots, b_d) \in \Z_{\geq 0}^d}{b_1 + 2 b_2 + \cdots + d b_d \leq W},\\
\lambda(W) &:= \setdef{(b_1,\ldots, b_d) \in \R_{\geq 0}^d}{b_1 + 2 b_2 + \cdots + d b_d \leq W}.
\end{align*}

\subsection{Counting points in $\Lambda(W)$}

Note that $\lambda(W)$ is just the standard simplex $S := \setdef{(b_1,\ldots, b_d) \in \R_{\geq 0}^d}{\sum b_i \leq 1}$ with the $i$-th direction scaled by $\tfrac{W}{i}$. It is well-known that $\vol(S) = \tfrac{1}{d!}$ and this immediately yields that 
\[
\vol(\lambda(W)) = \frac{W^d}{d!^2}\ . 
\]
Turns out, this is also a very good approximation of the number of points in $\Lambda(W)$ and the following lemma makes this precise. 

\begin{lemma}[relating discrete and continuous simplex]
  \label{lem:relating-discrete-continuous-simplex}
  \[
\vol(\lambda(W))\leq \abs{\Lambda(W)} \leq \vol(\lambda(W + d^2))\ .
  \]
\end{lemma}
\begin{proof}
  Take every point of $\vecb \in \Lambda(W)$ and consider the unit cube $\vecb + [0,1)^d$ (i.e., you may add any real number less than $1$ to each of the coordinates). Clearly $\lambda(W)$ is completely contained in the union of these cubes and that gives the inequality on the left. On the other hand, any point in such a cell is a vector of the form $(b_1',\ldots, b_d')$ with $b_i' < b_i + 1$. Hence, $\sum i b_i' \leq \sum i b_i + \sum i \leq W + d^2$. Hence, the  union of all these cells is contained within $\lambda(W + d^2)$. 
\end{proof}

\begin{corollary}
  \label{cor:simplex-size}
  If $W = \Omega(d^4)$, then $\abs{\Lambda(W)} = \vol(\lambda(W)) \cdot (1 \pm \tfrac{d^3}{W})$. 
\end{corollary}
\begin{proof}
Follows from the fact that the two sides of the inequality in the previous lemma are within a factor of $\inparen{1 + \tfrac{d^2}{W}}^d \approx 1 + \tfrac{d^3}{W}$. 
\end{proof}

\subsection{Expected value of a coordinate in $\Lambda(W)$}

We'll state the following fact we would need (with the proof provided in \cref{sec:deferred-proofs}) that relates the expected $\ell_1$ norm of a vector in $\Lambda(W)$.  

\begin{fact}
  \label{fact:expectation-of-l1-in-lattice}
For $W = \Omega(d^4)$, if $\mathcal{D}$ is the uniform distribution on $\Lambda(W)$, then 
\[
\E_{\vecb \sim \mathcal{D}}[b_i] = \frac{W}{i(d+1)} + O_d(1).
\]
Consequently, $\E_{\vecb \sim \mathcal{D}}[\sum b_i] = \frac{W \cdot H_d}{d+1} + O_d(1)$ where $H_d = \sum_{i=1}^d \tfrac{1}{i} \approx \ln d + \gamma + o(1)$ ($\gamma \approx 0.577..$ is the Euler-Mascheroni constant).
\end{fact}

\section{Algorithmic list-decoding: completing the proof}

We have seen the rough setup, so we will now proceed with the actual calculations. Let $A = k(1 + \epsilon)$, the number of agreements we are targeting. We will choose a multiplicity parameter $m = d^4$. 

\paragraph{Interpolation space ($\mathcal{M}$):} We will search for a nonzero polynomial $Q(X, Y_0, \ldots, Y_d)$ with support $X^a Y_0^{r} Y_1^{b_1} \cdots Y_d^{b_d}$ with 
\begin{itemize}\itemsep 0pt
\item $(1,k,k,\cdots, k)$-degree of $Q$ is less than $mA = mk(1 + \epsilon)$,
\item $(b_1,\ldots, b_d) \in \Lambda(W)$ where $W = \frac{md(1 + \tfrac{\epsilon}{2})}{H_d}$
\end{itemize}

By \cref{fact:expectation-of-l1-in-lattice}, note that the expected value of $b_1 + \cdots + b_d$ for this choice of $W$ is at most $m ( 1 + \tfrac{\epsilon}{2})$, which is smaller than $mA/k = m(1 + \epsilon)$. 

\paragraph{Number of indeterminates:} For ease of notation, we use $x_+$ to denote $\max(x, 0)$. If $\abs{\vecb} = \sum b_i$, then using \cref{fact:weighted-degree-mons} we get that the number of monomials that $Q$ is supported on is at least
\begin{align*}
\sum_{\vecb \in \Lambda(W)} \frac{(mA - k\abs{\vecb})_+^2}{2k} & = \abs{\Lambda(W)} \cdot \E_{\vecb \in \Lambda(W)} \frac{(mA - k\abs{\vecb})_+^2}{2k}\\
& \geq \abs{\Lambda(W)} \frac{(mA - k\E_{\vecb \sim \Lambda(W)}[\sum b_i])_+^2}{2k} & \text{(by Jensen on $x_+^2$)}\\
& \geq \abs{\Lambda(W)} \frac{(mk(1 + \epsilon) - mk(1 + \tfrac{\epsilon}{2}))_+^2}{2k}\\
& = \abs{\Lambda(W)} \cdot \frac{\epsilon^2}{8} \cdot m^2 k.
\end{align*}

\paragraph{The constraints ($\mathcal{C}$):} For each $i \in [n]$, consider the polynomial 
\begin{align*}
Q_i(T, E, Y_1, \ldots, Y_d) & = Q(\alpha_i + T, y_i + TY_1 - T^2 Y_2 + \cdots + (-1)^{d+1} T^d Y_d + E, Y_1, \ldots, Y_d)\\
& = \sum_{a,r, \vecc} q_{i, a, r, \vecc} \cdot T^a E^r Y_1^{c_1} \cdots Y_d^{c_d}
\end{align*}
where each $q_{i,a, r, \vecc}$ is a homogeneous linear combination of the coefficients of $Q$. For each $(a, r, \vecc)$ with $a + (d+1)r < m$, we add the constraint $q_{i,a, r, \vecc} = 0$. \\

Any monomial $Y_1^{b_1'} \cdots Y_d^{b_d'}$ obtained from $(y_i + TY_1 - T^2 Y_2 + \cdots + (-1)^{d+1} T^d Y_d + E)^r$ comes coupled with $T^{\sum i b_i'}$. The constraints above are only for monomials with $T$-degree bounded by $m$. Hence, we have that $c_1 + 2c_2 + \cdots + d c_d \leq W + m$ for all such monomials. Hence, the number of constraints we are adding is at most 
\begin{equation*}
n \cdot \frac{m^2 (1 + o(1))}{2(d+1)} \cdot \abs{\Lambda(W + m)} \leq n \cdot \frac{m^2}{d} \cdot \abs{\Lambda(W + m)} 
\end{equation*}

Therefore, a nonzero $Q$ satisfying the above constraints will exist if the number of indeterminates in $Q$ is more than the number of constraints we impose.

\begin{align*}
\frac{\text{\# indeterminates}}{\text{\# constraints}} & = \frac{\frac{\epsilon^2}{8} \cdot m^2 k \cdot \abs{\Lambda(W)}}{n \cdot \frac{m^2}{d} \cdot \abs{\Lambda(W + m)}}\\
& = \frac{\epsilon^2}{8} \cdot \text{rate} \cdot d \cdot \frac{\abs{\Lambda(W)}}{\abs{\Lambda(W + m)}} \approx  \frac{\epsilon^2}{8} \cdot \text{rate} \cdot d \cdot \frac{W^d}{(W+m)^d} & \text{($\because \; W \gg d^3$)}\\
& \geq \frac{\epsilon^2}{8} \cdot \text{rate} \cdot d \cdot \frac{1}{\exp(\tfrac{md}{W})} = \frac{\epsilon^2}{8} \cdot \text{rate} \cdot d \cdot \frac{1}{\exp(H_d \cdot \tfrac{2}{2+\epsilon})}\\
& \geq \frac{\epsilon^2}{8} \cdot \text{rate} \cdot d \cdot \frac{1}{\exp((\ln d + \gamma) \cdot \tfrac{2}{2+\epsilon})}\\
& \geq \frac{\epsilon^2}{8} \cdot \text{rate} \cdot d \cdot d^{-\frac{2}{2+\epsilon}} \cdot e^{-\gamma} \geq \frac{\epsilon^2}{16} \cdot \text{rate} \cdot d^{\frac{\epsilon}{2+\epsilon}}. 
\end{align*}

\noindent
Hence, if $d$ is set so that 
\[
d > \inparen{\frac{16}{\epsilon^2 \cdot \text{rate}}}^{\tfrac{2+\epsilon}{\epsilon}}, 
\]
we have more indeterminates than constraints and hence a nonzero $Q$ satisfying our constraints exists. The above proof also readily leads to a list-decoding algorithm that computes a nonzero $Q$ by finding a kernel vector of an $|\mathcal{C}| \times |\mathcal{M}|$ matrix where $|\mathcal{M}|, |\mathcal{C}| \leq n^{O_{\epsilon, \text{rate}}(1)}$, and then solves the differential equation via \cref{thm:kopparty-algorithm}. This completes the proof of \cref{thm:BCPZZ}. \\

\begin{algorithm}[H]
\caption{List-decoding $\RS_q[n,k]$}\label{alg:rs-decoder}
\KwIn{Received word $\vecy = (y_1,\dots,y_n)$ for evaluation points $(\alpha_1,\ldots, \alpha_n)$, and $\epsilon > 0$}
\KwOut{All $f$ of degree less than $k$ such that $\abs{\agr(f, \vecy)} \geq k(1 + \epsilon)$}
Choose parameters $d = \ceil{\inparen{\frac{16}{\epsilon^2 \cdot \nicefrac{k}{n}}}^{\nicefrac{3}{\epsilon}}}$, $m = d^4$ and $W = \ceil{\tfrac{md(1+\tfrac{\epsilon}{2})}{H_d}}$.\;
Set the linear system to find a nonzero $Q(X, Y_0,\ldots, Y_d) = \sum q_{a,r,\vecb} X^a Y_0^r Y_1^{b_1}\cdots Y_d^{b_d}$ with $(1,k,k,\cdots, k)$-degree less than $mk(1+\epsilon)$ and $\vecb \in \Lambda(W)$ that satisfy the following constraints: for all $i \in [n]$, all monomials in
\[
Q(\alpha_i + T, y_i + TY_1 - T^2 Y_2 + \cdots + (-1)^{d+1} T^d Y_d + E, Y_1,\ldots, Y_d)
\]
of the form $T^a E^r \mathbf{Y}^{\vecb}$ with $a + (d+1)r < m$ have a coefficient of zero. 
\;
Run the differential equation solver of \cite{Kopparty2015} to find all $f$ of degree less than $k$ such that 
\[
Q(X, f, f^{(1)}, \dots, f^{(d)}) = 0.
\]\;
\vspace{-\baselineskip}
\Return{all $f$'s from the previous line that agree with $\vecy$ on at least $k(1+\epsilon)$ locations.}\;
\end{algorithm}

\section{Proximity gaps}

An important related question is the notion of proximity gaps for these codes. Throughout this section, we shall focus on the instructive case of proximity gaps for lines (i.e., received word of the form $u + Z v$). The extension to more general statements is relatively straightforward.

\paragraph{Problem statement: } Suppose we are given two tables $u, v\colon \F_q\rightarrow \F_q$ such that for ``many'' values $\zeta \in \F_q$, there table $u + \zeta v$ agrees with some codeword of $\RS_q[n,k]$ on at least $A$ locations (i.e., there is some $f$ of degree less than $k$ such that $f$ and $u + \zeta v$ agree on at least $A$ evaluations among $\set{\alpha_1,\ldots, \alpha_n}$). Show that there are polynomials $g, h$ of degree less than $k$ and a set $S \subseteq [n]$ of size at least $A$ such that $g|_S = u|_S$ and $h|_S = v|_S$ (and consequently, $u + \zeta v$ agrees with $g + \zeta h$ on $S$, for any choice of $\zeta$). \\

This problem was first studied by Ben-Sasson, Carmon, Ishai, Kopparty, and Saraf~\cite{BenSassonCIKS2023}, who proved results when $A$ is at least the unique-decoding radius $\nicefrac{(n+k)}{2}$ or the Johnson radius $\sqrt{nk}$. The open question was then to understand the regime where $A$ is barely above the capacity of the code, i.e., $k(1+\epsilon)$ for arbitrarily small $\epsilon$. The main observation of Jeronimo~\cite{Jeronimo2026} is that the setup of \cite{BrakensiekCPZZ2026} and the differential equation solver of \cite{Kopparty2015} have enough algebraic structure to derive the above statement even for such small values of $A$. The overall sketch is quite natural, but the final analysis requires some tools from algebraic geometry. We will first provide the sketch, and defer the use of algebraic geometry as much as possible (a brief primer on the algebraic geometric concepts used in this exposition is provided in \cref{sec:alg-geom-primer}). 

\subsection{Overall sketch}

Let us fix some $\zeta_0 \in \F_q$ and a polynomial $f_0(X) \in \F_q[X]$ of degree less than $k$ with the property that $\abs{\agr(u + \zeta_0 v, f_0)} \geq A$. Let $S_{\zeta_0} := \agr(u + \zeta_0 v, f_0) = \setdef{i \in [n]}{(u + \zeta_0 v)(\alpha_i) = f_0(\alpha_i)}$. Fix any subset $T \subset S_{\zeta_0}$ of size $k$ and find the unique polynomials $g,h \in \F_q[X]$ of degree less than $k$ such that $g|_T = u|_T$ and $h|_T = v|_T$. Note that this also implies that $f_0 = g + \zeta_0 h$. 

Let $S := \agr(g, u) \cap \agr(h, v) = \setdef{i \in [n]}{g(\alpha_i) = u(\alpha_i), h(\alpha_i) = v(\alpha_i)}$. Clearly, $T \subseteq S \subseteq S_{\zeta_0}$ as any correlated agreement of $(g,h)$ with $(u,v)$ will translate to agreement of $g + \zeta_0 h = f_0$ with $u + \zeta_0 v$. If $S = S_{\zeta_0}$, we have the required correlated agreement we are after. Hence, we just need to address the situation that $S \subsetneq S_{\zeta_0}$. 

For any $i \in [n]$, let $H_i := \setdef{(\zeta,f)}{\deg(f) < k \text{ and } f(\alpha_i) = (u + \zeta v)(\alpha_i)}$, which is a affine hyperplane in $\F_q^{k+1}$.

\begin{lemma}
Suppose $i \in S_{\zeta_0} \setminus S$. Then, we have
\[
\bigcap_{j \in T \cup \{i\}} H_j = \set{(\zeta_0, f_0)}.
\]
\end{lemma}
\begin{proof}
  Note that for any $\zeta$, from our construction of $g$ and $h$, the unique polynomial of degree less than $k$ agreeing with $u + \zeta v$ on the set $T$ is the polynomial $g + \zeta h$. Observe that since $i \in S_{\zeta_0}\setminus S$, we get that either $u(\alpha_i) \neq g(\alpha_i)$ or $v(\alpha_i) \neq h(\alpha_i)$. Therefore, the polynomial
  \[
  (u(\alpha_i) - g(\alpha_i)) + Z(v(\alpha_i) - h(\alpha_i))
  \]
  is a nonzero linear function in $Z$ and hence there is at most one value of $Z$ that can make this equal to zero, and we already know that it is zero at $\zeta_0$ since $f_{0} = g + \zeta_0 h$. Therefore, at all $\zeta \neq \zeta_0$, we have $(u + \zeta v)(\alpha_i) \neq (g + \zeta h)(\alpha_i)$. 

  Since the size of $T$ is $k$, for every $\zeta$, the unique $f$ of degree less than $k$ that agrees with $(u+\zeta v)$ at every $\alpha \in T$ is $g+\zeta h$. But when $\zeta \neq \zeta_0$, the polynomial $g+\zeta h$ does not agree with $u + \zeta v$ at $\alpha_i$. Furthermore, for $\zeta = \zeta_0$, the polynomial $u+\zeta_0 v$ agrees with the degree $k$ polynomial $g + \zeta_0 h$ at $\alpha_j$ for $j \in T \cup \{i\}$, and since the size of $T$ is $k$, it is the unique such polynomial. 
\end{proof}

\noindent
We want to show that such $\zeta_0$'s (where $S_{\zeta_0} \neq S$) are quite rare, and we will do that by identifying such $\zeta_0$ through different means. This is where the setup of the interpolant system and the decoding algorithm actually enters the picture. Essentially, we can show that every $(\zeta, f) \in \F_q^{k+1}$ such that $\abs{\agr(u + \zeta v, f)} \geq A$ can be covered by a \emph{small} number of algebraic varieties whose dimension is at most $d+1$ and degree at most $n^{\poly(d)}$. This is a very general method that is applicable in such templates ``algebraic decoders'' such as \cite{WelchB1986,GemmellS1992,Sudan1997,GuruswamiS1999,BrakensiekCPZZ2026}. 

\begin{lemma}[informal]
  \label{lem:cover-by-few-small-varieties-informal}
  There is a set $B \subset \F_q$ of not-too-larger size, and irreducible varieties $\set{V_m}_{m \in [M]}$ (for $M$ not too large) such that $\dim V_m \leq d + 2$ and $\deg V_m \leq \poly(n)$ such that
  \[
  \setdef{(\zeta, f)}{\abs{\agr(u + \zeta v, f)}\geq A} \subseteq (B \times \F_q^k) \union  \bigcup_{m \in M} V_m.
  \]
\end{lemma}
More precisely, the decoding procedure can be used to show that either $\zeta$ belongs to an `obviously-bad' set $B$, or $(\zeta, f)$ is in the image of some rational function $\Psi_m: \F_q^{d+2} \rightarrow \F_q^{k+1}$ of not-too-large degree (and $V_m$ is just the smallest algebraic variety containing the image, and it is irreducible). 

The following is a standard fact from algebraic geometry.

\begin{fact} [{\cite[Theorem 2.2]{Dvir2012}}]
  Let $V$ be an irreducible variety of dimension $d$ and degree $D$. If $H$ is a hyperplane such that $\emptyset \neq H \cap V \neq V$, then if we write $H \cap V$ as a union of irreducible components $V_1 \union \cdots \union V_t$, then 
  \begin{itemize}\itemsep 0pt
    \item $\dim V_i = d - 1$ for all $i$,
    \item $\sum \deg V_i \leq D$. In particular, the number of components of $H \cap V$ is at most $D$. 
  \end{itemize}
\end{fact}

We are now ready to describe the \emph{pruning process} (which is inspired by the pruning procedure of Kopparty, Ron-Zewi, Saraf and Wootters~\cite{KoppartyRSW2023} and Tamo~\cite{Tamo2024}) to bound the number of $\zeta_0$ such that $S_{\zeta_0} \neq S$. We may assume that $\zeta_0 \notin B$ as $B$ is not-too-large. 

\paragraph{The pruning process:} Note that $\bigcap_{i \in S_{\zeta_0}} H_i = \set{(\zeta_0, f_0)}$. We build a sequence of irreducible varieties $C_0 \supsetneq  C_1 \supsetneq \cdots \supsetneq C_t = \set{(\zeta_0, f_0)}$ as follows. 
\begin{itemize}
\item Let $C_0$ be the one of the $V_m$'s that contain the point $(\zeta_0, f_0)$. 
\item For each $i > 0$, if $\dim C_i = 0$, we stop. Else, find choose $j_i \in S_{\zeta_0}$ such that $\emptyset \neq C_i \cap H_{j_i} \neq C_i$. Such a hyperplane must exist since $\cap_{i \in S_{\zeta_0}} H_i = \set{(\zeta_0, f_0 )}$ and we always maintain that $(\zeta_0, f_0) \in C_i$. Among the irreducible components of $C_i \cap H_{j_i}$, let $C_{i+1}$ be one that contains the point $(\zeta_0, f_0)$. 
\end{itemize}
Since the dimension keeps decreasing at each step, we have that $t \leq d+2$. Furthermore, the above gives a method to index into the point $(\zeta_0, f_0)$ by the tuple $(m_0, j_0, m_1, j_1, \cdots, j_{t-1}, m_{t-1})$ where $m_0 \in [M]$, each $j_i$ is the index of the hyperplane $H_{j_i}$ chosen at step $i$, and $m_i$ is the index of $C_{i+1}$ among irreducible components of $C_{i} \cap H_{j_i}$. The number of such tuples is at most $M \cdot (nD)^{O(d)}$ and thus the number of such $\zeta_0$'s is at most $|B| + M (nD)^{O(d)}$. 

If we set $L > |B| + M (nD)^{O(d)}$, we obtain \cref{thm:Jeronimo} which we restate below. 

\begin{theorem*}[(\cref{thm:Jeronimo} restated)] 
  Let $\rho, \epsilon > 0$, and let $k, n$ be large enough positive integers with $\nicefrac{k}{n} \geq \rho$.
  Consider the Reed-Solomon code $\RS_q[n,k]$ over a finite field $\F_q$ (of characteristic $p$) with $\nicefrac{k}{n} \geq \rho$ and $n \geq k(1+\epsilon)$ and $p \geq k$. 

  Suppose $u,v \colon E \rightarrow \F_q$ with the property that for at least $L$ values $\zeta \in \F_q$ there exists some degree $< k$ polynomial $f_\zeta$ such that $\abs{\agr(u + \zeta v, f_\zeta)} \geq A := k(1 + \epsilon)$. Then, there are polynomial $g,h$ of degree less than $k$ with correlated agreement of at least $A$. That is, there is an $S \subseteq [n]$ of size at least $A$ such that 
  \[
  \abs{\setdef{i \in [n]}{u(\alpha_i) = g(\alpha_i) \text{ and }v(\alpha_i) = h(\alpha_i)}} \geq k(1 + \epsilon)\, . 
  \]
\end{theorem*}

Thus, the task boils down to systematically going over the list-decoding setup and Kopparty's differential equation solver to get a bound on $|B|$, $D$ and $M$.

\begin{remark}
The above pruning already provides a bound for the number of codewords that have a $k(1 + \epsilon)$ agreement with any received word by a function that depends only on $n, k, \epsilon$, and independent of $p$. The same pruning process allows us to bound the number of exceptional $\zeta$'s in the proximity gap setting. 
\end{remark}

\paragraph{Recalling parameter settings: } Recall from the earlier proof that we will be choosing $m = d^4$ for $d = \inparen{\frac{16}{\epsilon^2 \cdot \text{rate}}}^{\tfrac{2+\epsilon}{\epsilon}}$. Thus, both $m$ and $d$ are constants independent of $n$.  

\subsection{Setting the interpolant system of equations over $\F(Z)$}

\begin{lemma}[Computing an interpolant over $\F(Z)$]
  \label{lem:interpolant-over-Fz}
There is a nonzero polynomial $Q_Z(X, Y_0,\ldots, Y_d) \in \F[Z, X, Y_0,\ldots, Y_d]$ with $\deg_Z(Q_Z) \leq \Delta = 2^{O(\poly(d))} \cdot \poly(n)$ such that for any $\zeta$ and polynomial $f(X)$  such that $\abs{\agr(u + \zeta v, f)} \geq A$, we have
\[
Q_\zeta(X, f, f^{(1)}, \ldots, f^{(d)}) \equiv 0. 
\]
\end{lemma}
\begin{proof}
Let us consider the received word as $u + Z v \in (\F[Z])^n$ with $y_i = u(\alpha_i) + Z v(\alpha_i)$ for $i \in [n]$. We follow the same procedure as earlier to now find a polynomial $Q_Z(X, Y_0,\ldots, Y_d) = \sum_{m \in \mathcal{M}} q_m \cdot m$ where $\mathcal{M}$ is the support of monomials allowed for $Q_Z$, and $q_m \in \F[Z]$ are the indeterminates. As earlier, we add constraints expecting $q_{i,a,r,\vecc}$ to be zero for $(i,a,r,\vecc) \in \mathcal{C}$. Since $Q_Z$ has $(1,k,\cdots, k)$-degree at most $mA$, each $q_{i,a,r,\vecc}$ is a polynomial in $Z$ of degree at most $mA = 2^{O(d)} \cdot n$. 

Therefore, we are looking for a kernel of an $\abs{\mathcal{C}}\times \abs{\mathcal{M}}$ matrix, with $\abs{\mathcal{C}} < \abs{\mathcal{M}} = 2^{\poly(d)} \cdot n$ and each of whose entries is a polynomial of degree at most $mA \leq \poly(d) \cdot n$. By Cramer's rule, we can obtain a nonzero polynomial $Q_Z(X, Y_0,\ldots, Y_d) \in \F[Z, X, Y_0,\ldots, Y_d]$ that is solution to the system of constraints with $\deg_Z Q_Z \leq \Delta := 2^{\poly(d)} \cdot n^2$. By a direct application of Cramer's rule the solution $Q$ might be a rational function in $Z$. By clearing the denominator, if needed, we can assume without loss of generality that the resulting $Q$ is a polynomial (and not a rational function) in $Z$. We may also assume without loss of generality that $Q_Z$ is \emph{primitive in $Z$} i.e., the coefficients of $Q_Z$ (as polynomials in $Z$) have no non-trivial common factor. 

Now, $Q_{\zeta}(X, Y_0,\ldots, Y_d)$ is a polynomial of $(1,k,\ldots, k)$-degree less than $mA$. Thus, like earlier, $R_f(X) := Q_{\zeta}(X, f, f^{(1)}, \ldots, f^{(d)})$ has $\alpha_i$ as a root of multiplicity $m$ whenever $f(\alpha_i) = y_i =(u + \zeta v)(\alpha_i)$ and this allows us to conclude that $R_f(X) = Q_\zeta(X, f, f^{(1)}, \ldots, f^{(d)}) \equiv 0$. 
\end{proof}

Hence, the interpolant $Q(Z, X, Y_0,\ldots, Y_d)$ is a polynomial with $\deg_Z Q \leq 2^{\poly(d)} \cdot n^2$, $\deg_X Q < mA \leq \poly(d) \cdot n$ and $\deg_{Y_i} Q \leq \tfrac{mA}{k} = m(1 + \epsilon) = \poly(d)$. 

\subsection{Kopparty's differential equation solver}

\newcommand{\TopVar}{\operatorname{TopVar}}

For a polynomial $P(Z, X, Y_0,\ldots, Y_d) \in \F[Z, X, Y_0, \ldots, Y_d]$, let $\TopVar(P)$ refer to the highest index $i \in \set{0,\ldots, d}$ such that $P$ depends non-trivially on $Y_i$. 

The following lemma is implicit in Kopparty's differential equation solver but we state it formally here. The proof of this is deferred to the Appendix (\cref{sec:degree-bounds-kopparty}). 

\begin{lemma}[\cite{Kopparty2015}]
  \label{lem:Kopparty-rational-fn}
  Let $P \in \F[Z, X, Y_0,\ldots, Y_d]$ be a nonzero polynomial with $\TopVar(P) = j \leq d$. Assume that $\operatorname{char}(\F) > \max(k, \deg_{Y_i}(P))$ for all $i$. There are rational functions $\Phi_{P,i}(Z, X, Y_0,\ldots, Y_d)$ defined  for $0 \leq i < k$ as 
  \[
  \Phi_{P,i}(Z, X, Y_0,\ldots, Y_d) = \begin{cases} Y_i & \text{if $0 \leq i \leq d$,}\\
    \frac{N_{P,i}(Z, X, Y_0,\ldots, Y_d)}{(\partial_{Y_j}P)^{2e}} & \text{if $d < i < k$},
  \end{cases}
  \]
with $\deg N_{P,i} \leq 2(i-j) \cdot \deg P$ and $e \leq 2(i - j)$ with the following property:
\begin{quote}
Suppose $\zeta \in \F$ is such that $\TopVar(P) = \TopVar(P|_{Z=\zeta})$, $f(X)$ is a polynomial of degree less than $k$, and $\gamma \in \F$ with $a_i = f^{(i)}(\gamma)$ for $0 < i < k$ such that 
\begin{itemize}\itemsep 0pt
  \item $P(\zeta, X, f, f^{(1)}, \ldots, f^{(d)}) = 0$,
  \item $(\partial_{Y_j} P)(\zeta, \gamma, a_0,\ldots, a_d) \neq 0$.
\end{itemize}
Then, for all $0 \leq i < k$,  
\[
  a_i = \Phi_{P,i}(\zeta, \gamma, a_0,\ldots, a_d). 
\]
Consequently, $f(X) = \sum_{i=0}^{k-1} \Phi_{P,i}(a_0,\ldots, a_d) \cdot (X - \gamma)^i$. 
\end{quote}
\end{lemma}

Intuitively, assuming some initial conditions on $\zeta, \gamma$ with respect to $f$, the first $d$ terms of the Taylor expansion of $f$ around the point $\gamma$ is sufficient to reconstruct $f$ entirely, and in a very explicit sense given by the rational functions $\Phi_{P,i}$. 

From the above $\Phi_{P,i}$, let us define a rational map $\Psi_{P, \gamma}$ as follows
\begin{align*}
\Psi_{P, \gamma}\colon (\zeta, a_0,\ldots, a_d) & \mapsto (\zeta, f_0,\ldots, f_{k-1})\\
\text{where } f(X) = \sum_{i=0}^{k-1} f_i x^i & = \sum_{i=0}^{k-1} \Phi_{P,i}(\zeta, \gamma, a_0,\ldots, a_d) \cdot (X - \gamma)^i.
\end{align*}
The above $\Psi$ is defined at all points in $\overline{\F}^{d+2}$ that ensures that the denominator $(\partial_{Y_j}P)(\zeta, \gamma, a_0,\ldots, a_d) \neq 0$. \\

Let us fix a polynomial $Q_Z(X, Y_0,\ldots, Y_d)$ obtained from \cref{lem:interpolant-over-Fz}, and a set $T \subseteq \F_p$ of at least $mA$ elements. Let $\partial_{Y}^{=*}(Q_Z)$ refer to the set of all polynomials that are derivatives of $Q_Z$ with respect to $Y$ variables of any order. Define the set $B$ of ``obviously-bad'' $\zeta$'s as
\[
B = \setdef{\zeta}{\begin{array}{c}\TopVar(P(Z, X, Y_0,\ldots, Y_d)) \neq \TopVar(P(\zeta, X, Y_0, \ldots, Y_d))\\
\text{for some non-constant } P(Z, X, Y_0,\ldots, Y_d) \in \partial_Y^{=*}(Q_Z)\end{array}}.
\]

\begin{claim}
  Fix any set $T \subseteq \F$ of size at least $mA$. 
  Suppose $\zeta \in \F \setminus B$ and $f(X)$ is a polynomial of degree less than $k$ such that we have $Q_Z(X, f, f^{(1)}, \ldots, f^{(d)}) = 0$. Then, $(\zeta, f)$ is in the image of $\Psi_{P,\gamma}$ for some $P \in \partial_Y^{=*}(Q)$ and $\gamma \in T$. 
\end{claim}
\begin{proof}
  Let $P(Z, X, Y_0,\ldots, Y_d)$ be the minimal degree non-zero polynomial in $\partial_{Y}^{=*}(Q_Z)$ with the property that $P(\zeta, X, f, f^{(1)}, \ldots, f^{(d)}) = 0$. From the minimality, we must have that $P'(X) = \partial_{Y_j}(P)(\zeta, X, f, f^{(1)}, \ldots, f^{(d)}) \neq 0$ where $j = \TopVar(P) = \TopVar(P_\zeta)$ (since $\zeta \notin B$). Since $P'(X)$ is a polynomial of degree less than $mA = |T|$, there is some $\gamma \in T$ such that $P'(\gamma) \neq 0$. Let $a_i = f^{(i)}(\gamma)$ for all $0 \leq i < k$. By \cref{lem:Kopparty-rational-fn}, we have that $\Psi_{P,\gamma}(\zeta, a_0,\dots, a_j) = (\zeta, f)$. 
\end{proof}

We need the following fact from algebraic geometry, which appears to be a standard consequence of \Bezout's theorem. We state it here without proof, but provide a rough proof-sketch in \cref{sec:alg-geom-primer}. 

\begin{fact}[image of a rational map on an open set (cf. {\cite[Lemma 5.4]{Jeronimo2026}})]
  \label{fact:image-of-rational-map}
  Let $\overline{\F}$ be the algebraic closure of $\F$. Suppose $\Psi(X_1,\ldots, X_d)$ be a rational map of the form 
\begin{align*}
\Psi & = \inparen{\frac{N_1}{G_1}, \ldots, \frac{N_n}{G_n}}
\end{align*}
where $N_i, G_i \in \overline{\F}[X_1,\ldots, X_d]$ with $\deg N_i, \deg G_i \leq \Delta$, that is well-defined on all points of $\overline{\F}^d$ where $G_1 \cdots G_n$ does not vanish. Then, the Zariski closure of $\operatorname{Image}(\Psi)$ is an algebraic variety $V$ such that 
\begin{itemize}\itemsep 0pt
\item $V$ is irreducible,
\item $\dim V \leq d$,
\item $\deg V \leq (n\Delta)^d$. 
\end{itemize}
\end{fact}

\noindent
We now have the formal version of \cref{lem:cover-by-few-small-varieties-informal}

\begin{corollary}[formal version of \cref{lem:cover-by-few-small-varieties-informal}]
  \label{cor:cover-by-few-small-varieties}
  Let $B$ be the set defined above, and $T\subseteq \F$ of at least $mA$ elements. Then, 
  \[
  \setdef{(\zeta, f)}{\abs{\agr(u + \zeta v, f)}\geq A} \subseteq (B \times \F^k) \union  \bigcup_{\substack{P\in \partial_Y^{=*}(Q)\\ \gamma \in T}} V_{P,\gamma}
  \]
  where $V_{P,\gamma} = \overline{\operatorname{Image}(\Psi_{P,\gamma})}$ defined above. Furthermore, each $V_{P,\gamma}$ is a variety of dimension at most $d+2$ and degree at most $2^{\poly(d)} \cdot n^{O(d)}$. 
\end{corollary}

\subsection{Bounding the final bad set of $\zeta$'s}

We now just need to plug in the bounds for $|B|$, $M$ and $D$ in \cref{lem:cover-by-few-small-varieties-informal} to compute $L$ for \cref{thm:Jeronimo}. 

Recall 
\[
B = \setdef{\zeta}{\begin{array}{c}\TopVar(P(Z, X, Y_0,\ldots, Y_d)) \neq \TopVar(P(\zeta, X, Y_0, \ldots, Y_d))\\
\text{for some non-constant } P(Z, X, Y_0,\ldots, Y_d) \in \partial_Y^{=*}(Q_Z)\end{array}}.
\]
Since $Q$ is a polynomial of $Z$-degree at most $n^2 \cdot 2^{O(\poly(d))}$, and $Y$-degree at most $\poly(d)$ there are at most $2^{\poly(d)}$ elements in $\partial_{Y}^{=*}(Q)$. Thus, the size of the set $B$ is at most $n^2 \cdot 2^{\poly(d)}$. 

\medskip

From \cref{cor:cover-by-few-small-varieties}, each $(\zeta, f)$ we are interested in with $\zeta \notin B$ are covered by one of the varieties $V_{P,\gamma}$. Thus, we have at most $M = \abs{\partial_{Y}^{=*}(Q_Z)} \cdot |T| = n \cdot 2^{\poly(d)}$ varieties of dimension $O(d)$ and degree at most $D = 2^{\poly(d)} \cdot n^{O(d)}$ each. Thus, if $p \gg n^{O(d^2)} \cdot 2^{\poly(d)}$, then even a small fraction of $\zeta$'s that witness $A$ agreements to $\RS[n,k]$ result in a correlated set of $A$ agreements. This complete the proof of \cref{thm:Jeronimo}.

{\let\thefootnote\relax
\footnotetext{\textcolor{\gitinfonotecolour}{\gitinfonote \easteregg}
}}
{\small 
  \bibliographystyle{prahladhurl}
  \bibliography{RScapacity-bib}
}

\appendix

\section{A brief primer on algebraic geometry}
\label{sec:alg-geom-primer}

We use this section to provide a very basic primer on algebraic geometry. The primary goal is not to provide an exhaustive, formal treatment, but rather to communicate the intuitive meaning of the concepts that would help digest the facts from algebraic geometry used in this write-up. \\

\noindent
Let $\F$ be an algebraically closed field.

\begin{definition}[Algebraic varieties: irreducibility, dimension and degree]
A subset $V$ of $\subseteq \F^n$ is said to be an algebraic variety if there is a finite set of polynomials $f_1,\ldots, f_m \in \F[x_1,\ldots, x_n]$ such that 
\[
V = \mathbb{V}(f_1,\ldots, f_m) := \setdef{\veca \in \F^n}{f_i(\veca) = 0 \text{ for all $i \in [m]$}}. 
\]
A variety is said to be \emph{irreducible} if it cannot be written as non-trivial finite union of two varieties. In other words, if $V = V_1 \union V_2$ for two algebraic varieties $V_1,V_2$, then either $V_1 = V$ or $V_2 = V$. 

The \emph{dimension of the variety $V$} (denoted by $\dim V$) is defined as the smallest $r$ such that a \emph{generic} affine subspace of co-dimension $r$ (or dimension, as standard vector space dimension, $n - r$) intersects $V$ at a finitely many points. The \emph{degree of the variety} (denoted by $\deg V$) is the finite number that results from such a generic intersection of a co-dimension $r$ affine subspace. 
\end{definition}

Since algebraic varieties are closed under finite unions and arbitrary intersections, we can define a topology on $\F^n$ by setting algebraic varieties as the \emph{closed sets}. The set $\F^n$ equipped with this topology (called the Zariski topology) is often referred to as $\mathbb{A}^n$. 

Under this topology, polynomial maps from $\mathbb{A}^d$ to $\mathbb{A}^n$ are continuous maps, and rational functions are continuous where defined (when the denominators don't vanish). The image of an irreducible space under a continuous map is always irreducible, and thus the image of rational maps (and thereby its closure too) are always irreducible. 

\medskip

Most quantitative bounds regarding degree of varieties follow from an application of \Bezout's theorem. The following is the form that is most useful to us. 

\begin{theorem}[\Bezout's theorem]
  \label{thm:bezout}
  Suppose $V = \mathbb{V}(f_1,\ldots, f_m) \subseteq \F^n$, with $\deg f_i = d_i$, then the sum of the degrees of its irreducible components is at most $\prod d_i$. 
\end{theorem}

\noindent
We now give an informal proof of \cref{fact:image-of-rational-map} (with a minor modification of making all denominators equal). 

\begin{fact}[Image of a rational map on an open set]
  Suppose $\Psi(X_1,\ldots, X_d):\mathbb{A}^d \setminus U \rightarrow \mathbb{A}^n$ be a rational map of the form 
  \begin{align*}
  \Psi & = \inparen{\frac{N_1}{G}, \ldots, \frac{N_n}{G}}
  \end{align*}
  where $N_i, G \in \F[X_1,\ldots, X_d]$ with $\deg N_i, \deg G \leq \Delta$ and $U = \mathbb{V}(G)$. Then, the Zariski closure of $\operatorname{Image}(\Psi)$ is an algebraic variety $V$ such that 
  \begin{itemize}\itemsep 0pt
  \item $V$ is irreducible,
  \item $\dim V \leq d$,
  \item $\deg V \leq \Delta^d$. 
  \end{itemize}
\end{fact}
\begin{proof}[Proof sketch]
  The set $\mathbb{A}^d \setminus U$ is irreducible, and hence its image under $\Psi$ (a continuous map) is also irreducible. Hence, the closure of the image is an irreducible algebraic variety. Since $\Psi$ maps $d$ coordinates into the space $\mathbb{A}^n$, we have $\dim V$ is at most $d$. 

  For the final point, let us assume for simplicity that $\dim V = d$ (otherwise, we have to work with a \emph{generic} $(\dim V)$-dimemsional subspace of $\mathbb{A}^d$ and repeat the argument). Let us take a generic affine subspace $L = \mathbb{V}(\ell_1, \cdots, \ell_d)$ of co-dimension $d$, where $\ell_i = a_{i0} + \sum_{j=1}^n a_{ij} y_j$, and intersect $V$ with it. For $i = 1,\ldots, d$, define
  \[
  P_i(x_1,\ldots, x_d) = c_{i0} G(\vecx) + \sum_{j=1}^n c_{ij} N_j(\vecx)
  \]
  and $Z = \mathbb{V}(P_1,\ldots, P_d)$, and note $\deg P_i \leq \Delta$ for all $i$. Since $\dim V = d$ and the $\ell_i$'s were generic, turns out it can be argued that each constraint $P_i = 0$ reduces the dimension of the space by one, and hence $Z$ will be finite. \cref{thm:bezout} shows that $|Z| \leq \Delta^d$. Each point of $V \cap L$ is certainly in $\Psi(Z)$. Therefore, 
  \[
  \deg V = \abs{V \cap L} \leq \abs{Z} \leq \Delta^n\qedhere
  \]
\end{proof}

\section{Deferred proofs}
\label{sec:deferred-proofs}

\begin{proof}[Proof of \cref{fact:expectation-of-l1-in-lattice}]
Consider the continuous relaxation $\lambda(W) = \setdef{(b_1,\ldots, b_d) \in \R_{\geq 0}^d}{\sum i b_i \leq W}$. Under the transformation $a_i = i \cdot b_i$, then $(a_1,\ldots, a_d)$ is an element of the standard simplex $S(W) = \setdef{(a_1,\ldots, a_d)}{\sum a_i \leq W}$. Let us add a slack variable $a_0$ so that $a_0 + \cdots + a_d = W$ and consider the set $\tilde{S}(W) = \setdef{(a_0,\ldots, a_d)}{\sum a_i = W}$. 

Clearly, $\tilde{S}(W)$ is symmetric over all the coordinates and hence $\E_{\veca \in \tilde{S}(W)}[a_i] = \frac{W}{d+1}$ for any $i$. Since the transformation from $a_i$'s to $b_i$'s is a coordinate-wise scaling, we get that 
\[
\E_{\vecb \in \lambda(W)} [b_i] = \E_{\veca \in \tilde{S}(W)}\insquare{\frac{a_i}{i}} = \frac{W}{i \cdot (d+1)}.
\]
Once again, we wish to relate the expectation over $\lambda(W)$ to the expectation over discrete set $\Lambda(W)$. As in \cref{lem:relating-discrete-continuous-simplex}, let us add a half-open cube $[0,1)^d$ to each point of $\Lambda(W)$. For the cube $(a_1,\cdots, a_d) + [0,1)^d$, the expectation of the $i$-th coordinate is just $a_i + \tfrac{1}{2}$. Hence, if $\Gamma = \Union_{\vecb \in \Lambda(W)} (\vecb + [0,1)^d)$, then 
$\lambda(W) \subseteq \Gamma \subseteq \lambda(W + d^2)$. Hence, 
\begin{align*}
\int_{\lambda(W)} x_i \, d x & \leq \int_{\Gamma} x_i \, d x \leq \int_{\lambda(W+d^2)} x_i \, d x\\
\text{i.e.,} \quad \vol(\lambda(W)) \cdot \E_{\vecb \in \lambda(W)}[b_i] &\leq \vol(\Gamma) \cdot \E_{\vecb \in \Gamma}[b_i] \leq \vol(\lambda(W + d^2)) \cdot \E_{\vecb \in \lambda(W + d^2)}[b_i]\\
\implies \vol(\lambda(W)) \cdot \frac{W}{(d+1)i} & \leq \inparen{\sum_{\vecb \in \Lambda(W)} b_i} + \tfrac{1}{2}\abs{\Lambda(W)} \leq \vol(\lambda(W + d^2)) \cdot \frac{(W+d^2)}{(d+1)i}.
\end{align*}
Dividing throughout by $\abs{\Lambda(W)}$, we get
\begin{align*}
\frac{W}{(d+1)\cdot i} \cdot (1 - O(\tfrac{d^3}{W})) & \leq \E_{\vecb \in \Lambda(W)}[b_i] + \frac{1}{2} \leq \frac{W + d^2}{(d+1)\cdot i} \cdot (1 + O(\tfrac{d^3}{W}))\\
\implies \E_{\vecb \in \Lambda(W)}[b_1 + \cdots + b_d] & = \frac{W}{(d+1)} \cdot \sum_{i=1}^d \frac{1}{i} \pm O_d(1) = \frac{W \cdot H_d}{(d+1)} \pm O_d(1)
\end{align*}
where $H_d  = \sum_{i=1}^d \frac{1}{i} = \ln d + \gamma + o(1)$.
\end{proof}

\section{Degree bounds for Kopparty's reconstruction algorithm}
\label{sec:degree-bounds-kopparty}

\begin{lemma}[\cref{lem:Kopparty-rational-fn} restated]
  Let $P \in \F[Z, X, Y_0,\ldots, Y_d]$ be a nonzero polynomial with $\TopVar(P) = j \leq d$. Assume that $\operatorname{char}(\F) > \max(k, \deg_{Y_i}(P))$ for all $i$. There are rational functions $\Phi_{P,i}(Z, X, Y_0,\ldots, Y_d)$ defined  for $0 \leq i < k$ as 
  \[
  \Phi_{P,i}(Z, X, Y_0,\ldots, Y_d) = \begin{cases} Y_i & \text{if $0 \leq i \leq d$,}\\
    \frac{N_{P,i}(Z, X, Y_0,\ldots, Y_d)}{(\partial_{Y_j}P)^{2e}} & \text{if $d < i < k$},
  \end{cases}
  \]
with $\deg N_{P,i} \leq 2(i-j) \cdot \deg P$ and $e \leq 2(i - j)$ with the following property:
\begin{quote}
Suppose $\zeta \in \F$ is such that $\TopVar(P) = \TopVar(P_{Z=\zeta})$, $f(X)$ is a polynomial of degree less than $k$, and $\gamma \in \F$ with $a_i = f^{(i)}(\gamma)$ for $0 < i < k$ such that 
\begin{itemize}\itemsep 0pt
  \item $P(\zeta, X, f, f^{(1)}, \ldots, f^{(d)}) = 0$,
  \item $(\partial_{Y_j} P)(\zeta, \gamma, a_0,\ldots, a_d) \neq 0$.
\end{itemize}
Then, for all $0 \leq i < k$,  
\[
  a_i = \Phi_{P,i}(\zeta, \gamma, a_0,\ldots, a_d). 
\]
Consequently, $f(X) = \sum_{i=0}^{k-1} \Phi_{P,i}(a_0,\ldots, a_d) \cdot (X - \gamma)^i$. 
\end{quote}
\end{lemma}

The crucial content of this lemma is in the fact that it gives a polynomial upper bound on the degrees of the numerator and denominator of $\Phi_{P,i}$. Such a bound is not entirely immediate from the reconstruction algorithm in \cite{Kopparty2015} since the reconstruction involves repeated composition of $P$ with the partial solution constructed so far (monomials of $f$ of low degree). Naively, this gives the impression that the degree of $\Phi_{P,i}$ is growing superpolynomially in $i$. We observe, using a slight modification of the reconstruction procedure, that this is not the case.

The rest of this section would be a proof of the above lemma. This is morally the same argument in Kopparty's paper but is a little more convenient to obtain the necessary degree bounds. For simplicity, let us assume that $\TopVar(P) = \TopVar(P_{\zeta}) = d$, where $P_\zeta = P(\zeta, X, Y_0, \ldots, Y_d)$. We have fixed some $\gamma \in T$ such that $(\partial_{Y_d}P_\zeta)(\gamma, a_0,\ldots, a_d) \neq 0$, where $a_i = f^{(i)}(\gamma)$. 

\paragraph{Expressing $a_{d+1}$ in terms of $a_0,\ldots, a_d$: } 
Note that for any $i$ we have $\tfrac{d}{dX} f^{(i)} = (i+1)f^{(i+1)}$. Since $Q(X, f, f^{(1)}, \ldots, f^{(d)}) = 0$, we have that 
\[
\frac{d}{dX}P_\zeta(X, f, f^{(1)}, \ldots, f^{(d)}) = 0
\]
as well. Let $\Delta(f)(X) := (f, f^{(1)}, f^{(2)}, \ldots, f^{(d)})$. 
\begin{align*}
0 & = \frac{d}{dX}P_\zeta(X, f, f^{(1)}, \ldots, f^{(d)})\\
  & = \partial_X P_{\zeta} (X, \Delta(f)) + \sum_{i=0}^{d-1} (i+1) \cdot \partial_{Y_i}(P_\zeta)(X, \Delta(F)) \cdot f^{(i+1)}\\
  & \quad \quad  + (d+1) f^{(d+1)} \cdot \partial_{Y_d} P_\zeta (X, \Delta(f))\\
\implies f^{(d+1)} & = \frac{-\inparen{\partial_X P_{\zeta} (X, \Delta(f)) + \sum_{i=0}^{d-1} (i+1) \cdot \partial_{Y_i}(P_\zeta)(X, \Delta(f)) \cdot f^{(i+1)}}}{(d+1) \cdot \partial_{Y_d} P_\zeta (X, \Delta(f))} 
\end{align*}
\noindent
Define the rational function $\Gamma_{P}(Z, X, Y_0,\ldots, Y_d)$ as 
\begin{align*}
\Gamma_{P}(Z, X, Y_0,\dots, Y_d) &= \frac{-\inparen{\partial_X P + \sum_{i=0}^{d-1} (i+1) \cdot \partial_{Y_i} P \cdot Y_{i+1}}}{(d+1) \cdot \partial_{Y_d} P}\\
\text{so that } f^{(d+1)} &:= \Gamma_{P}(Z, X, \Delta(f)). 
\end{align*}

Recall that we are trying to recover $f$ via its Taylor expansion around the point $\gamma$. Formally, $f(X) = \sum_{i}a_i(X-\gamma)^i$ and the goal is to obtain expressions for $a_{d+1}, \dots, a_{k-1}$ as functions of $a_0, a_1, \dots, a_d$ and $Z$. Thus given the formal expression for $f^{(d+1)}$ as above ($f^{(d+1)} = \Gamma_P(Z,X,\Delta(f))$), we can compute $a_{d+1}$ by setting $X$ to $\gamma$ in this expression (while making sure that this is well defined) and scaling the value obtained appropriately. The idea of the recovery algorithm here is to iterate on this: given the expression for $f^{(d+1)}$, we differentiate it $\ell$ times with respect to $X$ to obtain a rational expression for $f^{(d+1+\ell)}$ and then set $X$ to $\gamma$ in this expression to obtain $a_{d+1+\ell}$. The goal is to understand the bounds on the degree of the formal rational expressions for $f^{(d+1+\ell)}$ as $\ell$ increases, and this is done via a simple induction.  If we differentiate the expression for $f^{(d+1)}$ obtained above, via an application of chain rule and product rule of derivatives, the expression on the right hand side depends upon $f, f^{(1)}, \dots, f^{(d)}$, various derivatives of $P_{\zeta}$ and $f^{(d+1)}$. In this expression, we substitute the rational expression for $f^{(d+1)}$ obtained here so that the resulting expression only depends on $f, f^{(1)}, \dots, f^{(d)}$ and derivatives of $P_{\zeta}$. Staring at the form of this rational expression immediately gives us a sense of the increase in degree of the rational expression as we iterate this process. The rest of this section works out the formal details of this strategy. 

\paragraph{Building a derivation operator: } We will build a suitable ``derivation operator'' to help capture the structure of $\frac{d}{dX} h(Z, X, \Delta(f))$ for arbitrary $h(Z, X, Y_0,\ldots, Y_d)$. 
\begin{align*}
\mathcal{D} & := \partial_X + \sum_{i=0}^{d-1} (i+1) \cdot Y_{i+1} \cdot \partial_{Y_i} + (d+1) \cdot \Gamma_P \cdot \partial_{Y_d}
\end{align*}

\begin{observation}
For any $h(Z, X, Y_0,Y_1,\dots, Y_d)$, we have 
\[
\frac{d}{dX}(h(Z, X, \Delta(f))) = (\mathcal{D}h)(Z, X, \Delta(f)). 
\]
More generally, 
\[
\frac{d^t}{dX^t}(h(Z, X, \Delta(f))) = (\mathcal{D}^t h)(Z, X, \Delta(f)). 
\]
\end{observation}
\begin{proof} By the chain rule, 
\begin{align*}
\frac{d}{dX}(h(Z, X, \Delta(f))) & = (\partial_X h)(Z, X, \Delta(f)) + \sum_{i=0}^{d} (\partial_{Y_i} h)(Z, X, \Delta(f)) \cdot \frac{d}{dX} f^{(i)}\\
& = (\partial_X h)(Z, X, \Delta(f)) + \sum_{i=0}^{d} (\partial_{Y_i} h)(Z, X, \Delta(f)) \cdot (i+1) f^{(i+1)}\\
& = (\partial_X h)(Z, X, \Delta(f)) + \sum_{i=0}^{d-1} (\partial_{Y_i} h)(Z, X, \Delta(f)) \cdot (i+1) f^{(i+1)}\\
& \quad \quad  + (\partial_{Y_d} h)(Z, X, \Delta(f)) \cdot (d+1)f^{(d+1)}\\
& = (\partial_X h)(Z, X, \Delta(f)) + \sum_{i=0}^{d-1} (\partial_{Y_i} h)(Z, X, \Delta(f)) \cdot (i+1) f^{(i+1)}\\
& \quad \quad  + (\partial_{Y_d} h)(Z, X, \Delta(f)) \cdot (d+1) \cdot \Gamma_{P}(Z, X, \Delta(f))\\
& = (\mathcal{D}h) (Z, X, \Delta(f)). 
\end{align*}
The statement for higher derivatives follow by iterating the above $t$ times. 
\end{proof}

\begin{proposition}[Basic properties of the operator $\mathcal{D}$]
  \label{prop:degree-bounds-of-derivation}
For any $g,h \in \F(Z, X, Y_0,\ldots, Y_d)$, the following properties hold:
  \begin{itemize}\itemsep0pt
  \item $\mathcal{D}(g + h) = \mathcal{D}g + \mathcal{D}h$,
  \item $\mathcal{D}(gh) = h \mathcal{D}g + g\mathcal{D}h$,
  \item If $h \in \F[Z, X, Y_0, \ldots, Y_d]$ of degree at most $\Delta$, then $\mathcal{D}h$ is a rational function whose numerator has degree $\Delta + \deg P$ with denominator being a scalar multiple of $P'= \partial_{Y_d}(P)$. 
  \item More generally, for any $h \in \F[Z, X, Y_0,\ldots, Y_d]$, and any integer $e \geq 0$
  \[
  \mathcal{D}\inparen{\frac{h}{(P')^e}} = \frac{g}{(P')^{e+2}}
  \]
  for some polynomial $g \in \F[Z, X, Y_0,\ldots, Y_d]$ with $\deg g \leq \deg h + 2 \deg P'$. 
  \end{itemize}
\end{proposition}
\begin{proof}
The first two points follow readily from the definition of $\mathcal{D}$. The degree bounds also hold since the only introduction of denominators is from the $\Gamma_P \cdot \partial_{Y_d}$ term in $\mathcal{D}$ and the denominator is just $(d+1) P'$. For the last point, 
\begin{align*}
\mathcal{D}\inparen{\frac{h}{(P')^e}} & = \frac{\mathcal{D}(h)}{(P')^e} - \frac{e \cdot h \cdot \mathcal{D}(P')}{(P')^{e+1}}\\
& = \frac{P' \cdot \tilde{h} + e \cdot h \cdot \tilde{P} }{(P')^{e+2}}\\
\text{where } \mathcal{D}(P') & = \frac{\tilde{P}}{P'} \text{ and }\mathcal{D}(h) = \frac{\tilde{h}}{P'}.
\end{align*}
Then $g = P' \cdot \tilde{h} + e \cdot h \cdot \tilde{P}$, which is a polynomial of degree at most $\deg(h) + 2 \deg(P')$. 
\end{proof}

With the above derivation operator, we can obtain the rational functions computing $a_{d+t}$'s via repeated application of $\mathcal{D}$ to the function $h(Z, X, Y_0,\ldots, Y_d) := Y_d$.
\begin{align*}
(\mathcal{D}^t h)(Z, X, f, f^{(1)}, \ldots, f^{(d)}) & =  \frac{d^t}{d X^t} (h(Z, X, f, f^{(1)}, \ldots, f^{(d)})) = \frac{d^t}{d X^t} f^{(d)}\\
& = \frac{(d+t)!}{d!} \cdot f^{(d+t)}\qedhere
\end{align*}

\noindent
Evaluating the above expression for $f^{(d+t)}$ at $X = \gamma$, we get 
\[
a_{d+t} = f^{(d+t)}(\gamma) = \frac{d!}{(d+t)!} \cdot (\mathcal{D}^t Y_d)(\zeta, \gamma, a_0,\ldots, a_d). 
\]
The rational function $\Phi_{P, d+t}(Z, X, Y_0, \ldots, Y_d)$ is just $\tfrac{d!}{(d+t)!} \cdot (\mathcal{D}^t Y_d)$. Since $\deg Y_d = 1$, repeated applications of \cref{prop:degree-bounds-of-derivation} gives that 
\[
\mathcal{D}^t Y_d = \frac{N_{P,d+t}}{(P')^{e}}
\]
where $\deg N_{P,d+t} \leq 1 + (2t - 1) \deg P'  \leq 2t \deg P$ and $e \leq 2t-1$

\end{document}